\documentclass[letterpaper,11pt]{article}
\usepackage{times}
\usepackage{fullpage}
\usepackage{hyperref,microtype,pdfsync}
\usepackage{amsmath,amsfonts,amssymb,amsthm}
\usepackage{nicefrac}
\usepackage{url}
\usepackage{enumitem}
\usepackage{algorithm}
\usepackage[noend]{algorithmic}

\newcommand{\Q}{\ensuremath{\mathbb{Q}}}
\newcommand{\R}{\ensuremath{\mathbb{R}}}

\newcommand{\Z}{\ensuremath{\mathbb{Z}}}

\newcommand{\inner}[1]{\langle{#1}\rangle}

\newcommand{\abs}[1]{\lvert{#1}\rvert}
\newcommand{\absfit}[1]{\left\lvert{#1}\right\rvert}
\newcommand{\set}[1]{\{{#1}\}}
\newcommand{\setfit}[1]{\left\{{#1}\right\}}

\newcommand{\round}[1]{\lfloor{#1}\rceil}

\newcommand{\floor}[1]{\lfloor{#1}\rfloor}

\newcommand{\ceil}[1]{\lceil{#1}\rceil}

\newcommand{\length}[1]{\lVert{#1}\rVert}

\DeclareMathOperator{\poly}{poly}
\DeclareMathOperator{\polylog}{polylog}

\DeclareMathOperator*{\E}{E}

\theoremstyle{plain}            
\newtheorem{theorem}{Theorem}[section]
\newtheorem{lemma}[theorem]{Lemma}
\newtheorem{corollary}[theorem]{Corollary}
\newtheorem{proposition}[theorem]{Proposition}

\theoremstyle{definition}       
\newtheorem{definition}[theorem]{Definition}
\newtheorem{conjecture}[theorem]{Conjecture}

\theoremstyle{remark}           

\numberwithin{equation}{section}

\newif\ifnotes\notesfalse

\ifnotes

\newcommand{\dnote}[1]{{\bf (Daniel:} {#1}{\bf ) }}
\newcommand{\cnote}[1]{{\bf (Chris:} {#1}{\bf ) }}
\newcommand{\snote}[1]{{\bf (Santosh:} {#1}{\bf ) }}

\else

\newcommand{\dnote}[1]{}
\newcommand{\cnote}[1]{}
\newcommand{\snote}[1]{}

\fi

\def\Id{\mathrm{Id}}

\def\eps{\epsilon}

\def\st{~\mathrm{s.t.}~}
\DeclareMathOperator*{\conv}{conv}

\DeclareMathOperator*{\TVD}{d_{TV}}

\DeclareMathOperator*{\argmin}{arg\,min}
\DeclareMathOperator*{\cov}{cov}
\DeclareMathOperator{\SVP}{SVP}
\DeclareMathOperator{\CVP}{CVP}
\DeclareMathOperator*{\vol}{vol}
\def\nicefrac{\frac}

\newcommand{\pr}[2]{\left\langle #1, #2 \right\rangle }
\newcommand{\enc}[1]{\langle #1 \rangle}

\title{Enumerative Lattice Algorithms\\
  in Any Norm via M-Ellipsoid Coverings}

\author{Daniel Dadush\thanks{School of Computer Science, Georgia
    Institute of Technology.}
  \and
  Chris Peikert\footnotemark[1]
  \and
  Santosh Vempala\footnotemark[1]
}

\begin{document}

\maketitle

\begin{abstract}
  \ifnotes \begin{center}{\Huge{NOTES ARE ON}}\end{center} \fi
  We give a novel algorithm for enumerating lattice points in any convex
body, and give applications to several classic lattice problems,
including the Shortest and Closest Vector Problems (SVP and CVP,
respectively) and Integer Programming (IP).  Our enumeration technique
relies on a classical concept from asymptotic convex geometry known as
the \emph{M-ellipsoid}, and uses as a crucial subroutine the recent
algorithm of Micciancio and Voulgaris (STOC 2010) for lattice problems
in the $\ell_{2}$ norm.  As a main technical contribution, which may
be of independent interest, we build on the techniques of Klartag
(Geometric and Functional Analysis, 2006) to give an expected
$2^{O(n)}$-time algorithm for computing an M-ellipsoid for any
$n$-dimensional convex body.

As applications, we give deterministic $2^{O(n)}$-time and -space
algorithms for solving exact SVP, and exact CVP when the target point
is sufficiently close to the lattice, on $n$-dimensional lattices
\emph{in any (semi-)norm} given an M-ellipsoid of the unit ball.  In
many norms of interest, including all $\ell_{p}$ norms, an M-ellipsoid
is computable in deterministic $\poly(n)$ time, in which case these
algorithms are fully deterministic.  Here our approach may be seen as
a derandomization of the ``AKS sieve'' for exact SVP and CVP (Ajtai,
Kumar, and Sivakumar; STOC 2001 and CCC 2002).

As a further application of our SVP algorithm, we derive an expected
$O(f^*(n))^n$-time algorithm for Integer Programming, where $f^*(n)$
denotes the optimal bound in the so-called ``flatness theorem,'' which
satisfies $f^*(n) = O(n^{4/3} \polylog(n))$ and is conjectured to be
$f^{*}(n)=\Theta(n)$.  Our runtime improves upon the previous best of
$O(n^{2})^{n}$ by Hildebrand and K{\"o}ppe (2010).



\end{abstract}

\textbf{Keywords.}  Shortest/Closest Vector Problem, Integer
Programming, lattice point enumeration, M-ellipsoid.

\thispagestyle{empty}

\newpage

\setcounter{page}{1}

\section{Introduction}
\label{sec:introduction}

The Shortest and Closest Vector Problems (SVP and CVP, respectively)
on lattices are central algorithmic problems in the geometry of
numbers, with applications to Integer
Programming~\cite{lenstra83:_integ_progr_with_fixed_number_of_variab},
factoring polynomials over the rationals~\cite{lenstra82:_factor},
cryptanalysis
(e.g.,~\cite{odlyzko90:_rise_and_fall_of_knaps_crypt,DBLP:journals/joc/JouxS98,DBLP:conf/calc/NguyenS01}),
and much more.  (An $n$-dimensional \emph{lattice} $L$ is a discrete
additive subgroup of $\R^{n}$, and is generated as the set of integer
linear combinations of some basis vectors $b_{1}, \ldots, b_{k} \in
\R^{n}$, for some $k \leq n$.)  The SVP is simply: given a lattice $L$
represented by a basis, find a nonzero $v \in L$ such that
$\length{v}$ is minimized, where $\length{\cdot}$ denotes a particular
norm on $\R^{n}$.  The CVP is an inhomogeneous analogue of SVP: given
a lattice $L$ and a point $t \in \R^{n}$, find some $v \in L$ that
minimizes $\length{v - t}$.  In these problems, one often uses the
Euclidean ($\ell_{2}$) norm, but many applications require other norms
like $\ell_{p}$ or, most generally, the semi-norm defined by a convex
body $K \ni 0$ as $\length{x}_{K} = \inf \set{ r \geq 0 : x \in rK }$.
Indeed, general (semi-)norms arise quite often in the study of
lattices; for example, the ``flatness theorem'' in Integer Programming
--- which states that every lattice-free convex body has lattice width
bounded by a function of the dimension alone --- is a statement about
SVP in general norms.

Much is known about the computational complexity of SVP and CVP, in
both their exact and approximation versions.  On the negative side,
SVP is NP-hard (in $\ell_{2}$, under randomized reductions) to solve
exactly, or even to approximate to within any constant
factor~\cite{DBLP:conf/stoc/Ajtai98,DBLP:journals/jcss/CaiN99,DBLP:journals/siamcomp/Micciancio00,DBLP:journals/jacm/Khot05}.
Many more hardness results are known for other $\ell_{p}$ norms and
under stronger complexity assumptions than P $\neq$ NP (see,
e.g.,~\cite{emde81:_anoth_np,DBLP:journals/tcs/Dinur02,DBLP:conf/stoc/RegevR06,DBLP:conf/stoc/HavivR07}).
CVP is NP-hard to approximate to within $n^{c/\log \log n}$ factors
for some constant $c >
0$~\cite{DBLP:journals/jcss/AroraBSS97,DBLP:journals/combinatorica/DinurKRS03,DBLP:journals/tcs/Dinur02},
where $n$ is the dimension of the lattice.  Therefore, we do not
expect to solve (or even closely approximate) these problems
efficiently in high dimensions.  Still, algorithms providing weak
approximations or having super-polynomial running times are the
foundations for the many applications mentioned above.

The celebrated LLL algorithm~\cite{lenstra82:_factor} and
variants~\cite{DBLP:journals/tcs/Schnorr87} give $2^{n /
  \text{polylog}(n)}$ approximations to SVP and CVP in $\ell_{2}$, in
$\poly(n)$ time.  For exact SVP and CVP in the $\ell_{2}$ norm,
Kannan's
algorithm~\cite{kannan87:_minkow_convex_body_theor_and_integ_progr}
gives a solution in deterministic $2^{O(n \log n)}$ time and
$\poly(n)$ space.  This performance remained essentially unchallenged
until the breakthrough randomized ``sieve'' algorithm of Ajtai, Kumar,
and Sivakumar~\cite{DBLP:conf/stoc/AjtaiKS01}, which provides a
$2^{O(n)}$-time and -space solution for exact SVP; moreover, the
algorithm generalizes straightforwardly to $\ell_{p}$ and other
norms~\cite{DBLP:conf/icalp/BlomerN07,DBLP:conf/fsttcs/ArvindJ08}. For
CVP, in a sequence of works
~\cite{DBLP:conf/coco/AjtaiKS02,DBLP:conf/icalp/BlomerN07,DBLP:conf/fsttcs/ArvindJ08}
it was shown that a modified version of the AKS sieve can approximate
CVP in any $\ell_p$ norm to within a $(1+\epsilon)$ factor in time and
space $(1/\epsilon)^{O(n)}$ for any $\epsilon > 0$. Furthermore, these
algorithms can solve CVP exactly in $2^{O(n)}$ time as long as the
target point is ``very close'' to the lattice. It is worth noting that
the AKS sieve is a \emph{Monte Carlo} algorithm: while the output
solution is correct with high probability, it is not guaranteed.

In a more recent breakthrough, Micciancio and
Voulgaris~\cite{DBLP:conf/stoc/MicciancioV10} gave a
\emph{deterministic} $2^{O(n)}$-time (and space) algorithm for exact
SVP and CVP in the $\ell_{2}$ norm, among many other lattice problems
in NP.  Interestingly, their algorithm works very differently from the
AKS sieve, by computing an explicit description of the Voronoi cell of
the lattice.  (The Voronoi cell is the set of all points in $\R^{n}$
that are closer to the origin than to any other lattice point.)  In
contrast to the AKS sieve, however, the algorithm
of~\cite{DBLP:conf/stoc/MicciancioV10} appears to be quite specialized
to $\ell_{2}$ (or any norm defined by an ellipsoid, simply by applying
a linear transformation).  This is in part because in $\ell_{2}$ the
Voronoi cell is convex and has $2^{O(n)}$ facets, but in general norms
this is not the case.  A main problem left open
in~\cite{DBLP:conf/stoc/MicciancioV10} was to find deterministic
$2^{O(n)}$-time algorithms for lattice problems in $\ell_{p}$ and
other norms.

\subsection{Results and Techniques}
\label{sec:results-techniques}

Our main contribution is a novel algorithm for enumerating lattice
points in any convex body.  It uses as a crucial subroutine the
Micciancio-Voulgaris (MV) algorithm
\cite{DBLP:conf/stoc/MicciancioV10} for the $\ell_{2}$ norm that
enumerates lattice points in an ellipsoid, and relies on a classical
concept from asymptotic convex geometry known as the M-ellipsoid.
This connection between lattice algorithms and convex geometry appears
to be a fertile direction for further research.

For a lattice $L$ and convex body $K$ in $\R^n$, let $G(K,L)$ be the
largest number of lattice points contained in any translate of $K$, i.e., 
\begin{equation}
  \label{eq:G-K-L}
  G(K,L) = \max_{x \in \R^{n}} \abs{(K+x) \cap L}.
\end{equation}
Our starting point is the following guarantee on the enumeration of $K
\cap L$.\footnote{For simplicity, throughout this introduction the
  claimed running times will omit polynomial factors in the lengths of
  the algorithms' inputs, which are represented in the usual way.}

\begin{theorem}[Enumeration in convex bodies, informal]
  \label{thm:body-lat-enum-informal}
  Given any convex body $K \subseteq \R^n$ along with an
  \emph{M-ellipsoid} $E$ of $K$, and any $n$-dimensional lattice $L
  \subseteq \R^n$, the set $K \cap L$ can be computed in deterministic
  time $G(K,L) \cdot 2^{O(n)}$.
\end{theorem}

As we describe later, an M-ellipsoid $E$ of a convex body $K \subseteq
\R^n$ is an ellipsoid with roughly the same `size' and `shape' as $K$.
We will show that it can generated in randomized $\poly(n)$ time with
high probability, and verified in deterministic $2^{O(n)}$ time, and
hence can always be computed in expected $2^{O(n)}$ time. Moreover, in
many specific cases of interest, such as the unit ball of any
$\ell_{p}$ norm, an M-ellipsoid is deterministically computable in
$\poly(n)$ time.

Our enumeration algorithm is at the core of the following
applications.  We begin with the Shortest Vector Problem in {\em any}
``well-centered'' semi-norm.\footnote{``Well-centered'' means that
  $\vol(K \cap -K) \geq 4^{-n} \vol(K)$; this clearly holds for
  centrally symmetric $K$, which corresponds to a standard norm.  It
  also holds for any convex body $K$ with centroid at or very near the
  origin.}

\begin{theorem}[SVP in any (semi-)norm, informal]
  \label{thm:svp-deterministic-informal}
  There is a deterministic $2^{O(n)}$-time (and -space) algorithm
  that, given any well-centered $n$-dimensional convex body $K$ and
  an M-ellipsoid $E$ of $K$, solves SVP
  exactly on any $n$-dimensional lattice $L$ in the semi-norm
  $\length{\cdot}_{K}$ defined by $K$.
\end{theorem}

Besides being a novel algorithm, the improvement over previous
approaches is in the generalization to (semi-)norms defined by
arbitrary convex bodies, the use of much less randomness (if any), and
in having a Las Vegas algorithm whose output is guaranteed to be
correct.

We get a similar algorithm for the Closest Vector Problem, but its
complexity grows with the distance from the target point to the
lattice.

\begin{theorem}[CVP in any (semi-)norm, informal]
  \label{thm:cvp-deterministic-informal}
  There is a deterministic algorithm that, given any well-centered
  $n$-dimensional convex body $K$ and an M-ellipsoid $E$ of $K$,
  solves CVP exactly on any $n$-dimensional lattice $L$ in the
  semi-norm $\length{\cdot}_{K}$ defined by $K$, in $(1+2\alpha)^{n}
  \cdot 2^{O(n)}$ time and space, provided that the distance from the
  query point $x$ to $L$ is at most $\alpha$ times the length of the
  shortest nonzero vector of $L$ (under $\length{\cdot}_{K}$).
\end{theorem}

A main motivation of our work is to develop more powerful tools for
solving Integer Programming.  We note that solving IP reduces to
solving CVP in any well-centered semi-norm: to decide if $K \cap L
\neq \emptyset$, first approximate the centroid $b$ of $K$, then solve
CVP with respect to the well-centered body $K-b$ on lattice $L$ and
target point $b$.  Then $K \cap L \neq \emptyset$ if and only if there
exists $y \in L$ such that $\length{y-b}_{K-b} \leq 1$.  However,
unless we have a bound on the ratio $\alpha$ from the above theorem,
we may not get a satisfactory guarantee on the running time of our CVP
algorithm in this setting.

For the general case, we can still get an unqualified improvement in
the state of the art for IP using our SVP algorithm for general norms.

\begin{theorem}[Integer Programming, informal]
  \label{thm:faster-ip-informal}
  There exists a randomized algorithm that, given a convex body $K
  \subseteq \R^{n}$ and an $n$-dimensional lattice $L \subset \R^{n}$,
  either decides that $K \cap L = \emptyset$ or returns a point $y \in
  K \cap L$ in expected $O(f^*(n))^n$ time, where $f^{*}(n)$ is the
  optimal bound for the ``flatness theorem.''
\end{theorem}

The flatness theorem, a fundamental result in the geometry of numbers,
says that every lattice-free convex body has lattice width bounded by
a function of the dimension alone (see Equation~\eqref{eq:flatness}
for a precise statement).  As first noticed by
Lenstra~\cite{lenstra83:_integ_progr_with_fixed_number_of_variab}, it
suggests a recursive algorithm for IP that uses a subroutine for
finding good flatness directions.  Finding an optimal flatness
direction directly reduces to solving an SVP in a general norm, which
was solved only approximately in previous refinements of Lenstra's
algorithm.  The above is therefore an essentially ``optimal''
Lenstra-type algorithm with respect to the classical analysis.

Using the current best known bounds on $f^*(n)$~\cite{Bana99,R00}, our
IP algorithm has a main complexity term of order $O(n^{4/3} \log^c
n)^n$.  This improves on the previous fastest algorithm of Hildebrand
and K{\"o}ppe~\cite{arxiv/HildebrandK10}
which gives a leading complexity term of $O(n^2)^n$; the previous best
before that is due to
Kannan~\cite{kannan87:_minkow_convex_body_theor_and_integ_progr} and
achieves a leading complexity term of $O(n^{2.5})^n$.  It is
conjectured that $f^*(n) = \Theta(n)$~\cite{Bana99}, and this would
give a bound of $O(n)^n$ for IP.

In the rest of this introduction we give an overview of our
enumeration technique and its application to SVP, CVP, and IP.

\paragraph{Enumeration via M-ellipsoid coverings.}

We now explain the main technique underlying
Theorem~\ref{thm:body-lat-enum-informal} (enumeration of lattice
points in a convex body $K$).  The key concept we use is a classical
notion from asymptotic convex geometry, known as the
\emph{M-ellipsoid}.  An M-ellipsoid $E$ for a convex body $K$ has the
property that $2^{O(n)}$ copies (translates) of $E$ can be used to
cover $K$, and $2^{O(n)}$ copies of $K$ suffice to cover~$E$.  The
latter condition immediately implies that
\begin{equation}
  \label{eq:GKL-covering}
  G(E,L) \leq 2^{O(n)} \cdot G(K,L).
\end{equation}
Using the former condition, enumerating $K \cap L$ therefore reduces
to enumerating $(E+t) \cap L$ for at most $2^{O(n)}$ values of $t$
(and keeping only those lattice points in $K$), which can be done in
deterministic $2^{O(n)} \cdot G(E,L)$ time by (an extension of) the MV
algorithm~\cite{DBLP:conf/stoc/MicciancioV10}.

The existence of an M-ellipsoid for any convex body $K$ was
established by Milman~\cite{M86,MP00}, and there are now multiple
proofs.  Under the famous \emph{slicing conjecture}~\cite{B86}, an
appropriate scaling of $K$'s \emph{inertial ellipsoid} (defined by the
covariance matrix of a uniform random point from $K$) is in fact an
M-ellipsoid.  When $K$ is an $\ell_{p}$ ball, an M-ellipsoid is simply
the scaled $\ell_{2}$ ball $n^{1/2-1/p} \cdot B_{2}^{n}$.

For general convex bodies $K$, we give an algorithm for computing an
M-ellipsoid of $K$, along with a covering by copies of the ellipsoid.
Under the slicing conjecture, the former task is straightforward:
simply estimate the covariance matrix of $K$ using an algorithm for
sampling uniformly from a convex body (e.g.,~\cite{DyerFK91}).  To
avoid assuming the slicing conjecture, we use an alternative proof of
M-ellipsoid existence due to Klartag~\cite{K06}.  The resulting
guarantees can be stated as follows.

\begin{theorem}[M-ellipsoid generator, informal]
  \label{thm:M-gen-informal}
  There is a polynomial-time randomized algorithm that with high
  probability computes an M-ellipsoid $E$ of a given $n$-dimensional
  convex body $K$.\footnote{We thank Bo'az Klartag for suggesting to
    us that the techniques in~\cite{K06} could be used to
    algorithmically construct an M-ellipsoid.}
\end{theorem}

\begin{theorem}[M-ellipsoid covering algorithm, informal]
  \label{thm:M-cover-informal}
  Given an ellipsoid $E$ and convex body $K$, there is a deterministic
  $2^{O(n)}$-time algorithm which certifies that $E$ is an M-ellipsoid
  of $K$, and if so returns a covering of $K$ by $2^{O(n)}$ copies of
  $E$.\footnote{Gideon Schechtman suggested a construction of the
    covering using parallelepiped tilings.}
\end{theorem}

Combining these two theorems, we get an expected $2^{O(n)}$-time
algorithm that is guaranteed to output an M-ellipsoid and its implied
covering for any given convex body $K$.  It is an interesting open
problem to find a \emph{deterministic} $2^{O(n)}$-time algorithm.  We
note that deterministic algorithms must have complexity
$2^{\Omega(n)}$, since an M-ellipsoid gives a $2^{O(n)}$ approximation
to the volume of $K$, and such an approximation is known to require
$2^{\Omega(n)}$ time when $K$ is specified by an oracle \cite{BF87}.

\vspace{-6pt}
\paragraph{Shortest and Closest Vector Problems.}

Here we outline our deterministic $2^{O(n)}$-time algorithm for SVP in
any norm defined by a symmetric convex body~$K$, given an M-ellipsoid
of $K$.  (Well-centered semi-norms are dealt with similarly.)  For
instance, as noted above the scaled $\ell_{2}$ ball $E_{p}=n^{1/2-1/p}
\cdot B_{2}^{n}$ is an M-ellipsoid for any $\ell_{p}$ ball $K =
B_{p}^{n}$.  Moreover, a good covering of~$B_{p}^{n}$ by $E_{p}$ is
straightforward to obtain: for $p \geq 2$, just one copy of $E_{p}$
works (since $B_{p}^{n} \subseteq E_{p}$), while for $1 \leq p < 2$,
we can cover $B_{p}^{n}$ by a tiling of $E_{p}$'s axis-aligned
inscribed cuboid.

Let $L$ be an $n$-dimensional lattice, and let $\lambda_1 =
\lambda_1(K,L)$ be the length of its shortest vector under
$\length{\cdot}_{K}$.  We can assume by rescaling that $1/2 <
\lambda_{1} \leq 1$, so $K$ contains an SVP solution.  Our algorithm
simply enumerates all nonzero points in $K \cap L$ (using
Theorem~\ref{thm:body-lat-enum-informal}), and outputs one of the
shortest.  For the running time, it suffices to show that $G(K,L) \leq
2^{O(n)}$, which follows by a simple packing argument: for any $x \in
\R^{n}$, copies of $\tfrac14 K$ centered at each point in $(K+x) \cap
L$ are pairwise disjoint (because $\lambda_{1} > 1/2$) and contained
in $\tfrac54 K + x$, so $\abs{(K+x) \cap L} \leq \vol(\tfrac54 K) /
\vol(\tfrac14 K) = 5^{n}$.

For CVP with target point $x$, the strategy is exactly the same as
above, but we use a scaling $dK$ so that $(dK - x) \cap L \neq
\emptyset$ and $(\tfrac{d}{2}K - x) \cap L = \emptyset$ (i.e., $d$ is
a $2$-approximation of the distance from $x$ to $L$).  In this case,
the packing argument gives a bound of $G(dK,L) \leq
(1+2d/\lambda_{1})^{n}$.

In retrospect, the above algorithms can be seen as a derandomization
(and generalization to semi-norms) of the AKS sieve-based algorithms
for exact SVP in general norms, and exact CVP in $\ell_{p}$
norms~\cite{DBLP:conf/stoc/AjtaiKS01,DBLP:conf/coco/AjtaiKS02,DBLP:conf/icalp/BlomerN07,DBLP:conf/fsttcs/ArvindJ08},
with matching running times (up to $2^{O(n)}$ factors).  Specifically,
our algorithms deterministically enumerate all lattice points in a
convex region, rather than repeatedly sampling until all such points
are found with high probability.  However, we do not know whether our
techniques can derandomize the $(1+\epsilon)$-approximate CVP
algorithms
of~\cite{DBLP:conf/coco/AjtaiKS02,DBLP:conf/icalp/BlomerN07} in
asymptotically the same running time.

\vspace{-6pt}
\paragraph{Integer Programming.}

Our algorithm for Integer Programming (finding a point in $K \cap L$,
if it exists) follows the basic outline of all algorithms since that
of Lenstra~\cite{lenstra83:_integ_progr_with_fixed_number_of_variab}.
It begins with two pre-processing steps: one to refine the basis of
the lattice, and the other to find an ellipsoidal approximation of
$K$.  If the ellipsoid volume is sufficiently small compared to the
lattice determinant, then we can directly reduce to a
lower-dimensional problem.  The main step of the algorithm (and
Lenstra's key insight, refined dramatically by
Kannan~\cite{kannan87:_minkow_convex_body_theor_and_integ_progr}) is
to find a direction along which the lattice width of $K$ is small.
Given such a direction, we recurse on the lattice hyperplanes
orthogonal to this direction that intersect $K$, thus reducing the
dimension of the problem by one.

In previous work, a small lattice-width direction was found by
replacing $K$ by an ellipsoid $E$ \emph{containing}~$K$, then solving
SVP in the norm defined by the dual ellipsoid $E^{*}$ on the dual
lattice $L^*$.  Here we instead use our SVP algorithm for general
norms, solving it directly for the norm induced by $(K-K)^*$ on
$L^{*}$.  This refinement allows us to use the best-known bounds on
$f^{*}(n)$ (from the flatness theorem) for the number of hyperplanes
on which we have to recurse.

\subsection{Organization}
\label{sec:organization}

The remainder of the paper is organized as follows.  In
Section~\ref{sec:basic-concepts} we recall basic concepts from convex
geometry that are needed to understand our M-ellipsoid algorithms.  In
Section~\ref{sec:m-ellipsoid-covering} we give the M-ellipsoid
construction (formalizing Theorems~\ref{thm:M-gen-informal}
and~\ref{thm:M-cover-informal}).  In Section~\ref{sec:lattice-algs} we
formalize our enumeration technique
(Theorem~\ref{thm:body-lat-enum-informal}) and apply it to give
algorithms for SVP, CVP and IP.  Appendix~\ref{sec:m-ellipsoid-proofs}
contains the proofs of correctness for our M-ellipsoid construction,
and Appendix~\ref{sec:additional-background} contains supporting
technical material.

\section{Convex Geometry Background}
\label{sec:basic-concepts}

\paragraph{Convex bodies.}

$K \subseteq \R^n$ is a convex body if $K$ is convex, compact and
full-dimensional.  We say that a body is centrally symmetric, or
$0$-symmetric, if $K = -K$.

For sets $A,B \in \R^n$ we define the Minkowski sum of $A$ and $B$ as
\begin{equation}
  A + B = \set{x + y: x \in A, y \in B}.
\end{equation}
For a vector $t \in \R^n$, we define $t + A = \set{t} + A$ for
notational convenience.

Let $K \subseteq \R^n$ be a convex body such that $0 \in K$. We define
the gauge function, or Minkowski functional, of $K$ as
\begin{equation}
  \length{x}_K = \inf \set{r \geq 0: x \in rK}, \quad x \in \R^n.
\end{equation}

From classical convex analysis, we have that the functional
$\length{\cdot}_K$ 
is a semi-norm, i.e., it satisfies the triangle inequality and
$\length{tx}_{K} = t\length{x}_{K}$ for $t \geq 0$, $x \in \R^{n}$.
If $K$ is centrally symmetric, then $\length{.}_K$ is a norm in the
usual sense.

The \emph{polar} (or \emph{dual}) body $K^{*}$ is defined as
\begin{equation}
  K^* = \set{x \in \R^n: \forall y \in K, \; \pr{x}{y} \leq 1}.
\end{equation}
A basic result in convex geometry is that $K^*$ is convex and that
$(K^{*})^{*} = K$.

Define the $\ell_{p}$ norm on $\R^{n}$ as
\begin{equation}
  \length{x}_p = \left(\sum_{i=1}^n |x_i|^p \right)^{\nicefrac{1}{p}}.
\end{equation}
For convenience we write $\length{x}$ for $\length{x}_{2}$.  Let
$B_p^n = \set{x \in \R^n: \length{x}_p \leq 1}$ denote the $\ell_p$
ball in $\R^n$.  Note from our definitions that $\length{x}_{B_p^n} =
\length{x}_p$ for $x \in \R^n$.

For a positive definite matrix $A \in \R^{n \times n}$, we define the
inner product with respect to $A$ as
\begin{equation}
  \pr{x}{y}_A = x^t A y \quad x,y \in \R^n.
  \label{def:ip}
\end{equation}
We define the norm generated by $A$ as $\length{x}_A =
\sqrt{\pr{x}{x}_A} = \sqrt{x^t A x}$.  For a vector $a \in \R^{n}$, we
define the ellipsoid $E(A,a) = \set{x \in \R^n: \length{x-a}_A \leq
  1}$.  For convenience we shall let $E(A) = E(A,0)$.  Note that with
our notation, $\length{x}_A = \length{x}_{E(A)}$.  The volume of an
ellipsoid $E(A,a)$ is given by the formula
\begin{equation}
  \vol(E(A,a)) = \vol(E(A)) = \vol(B_2^n) \cdot \sqrt{\det(A^{-1})}.
  \label{eq:ell-vol-form}
\end{equation}
Lastly, an elementary computation gives the useful fact that $E(A)^* =
E(A^{-1})$.

We define the \emph{centroid} (or \emph{barycenter}) $b(K) \in \R^{n}$
and \emph{covariance} matrix $\cov(K) \in \R^{n \times n}$ as
\begin{align*}
  b(K) &= \int_{K} \frac{x\, dx}{\vol(K)} & \cov(K) &= \int_K
  (x-b(K))(x-b(K))^{t} \frac{dx}{\vol(K)}.
\end{align*}
We note that $\cov(K)$ is always positive definite and symmetric.  The
\emph{inertial ellipsoid} of $K$ is defined as $E_{K} = E(\cov(K)^{-1})$.
The \emph{isotropic constant} of $K$ is
\begin{equation}
  L_K = \det(\cov(K))^{\nicefrac{1}{2n}}/\vol(K)^{\nicefrac{1}{n}}.
\end{equation}

A major open conjecture in convex geometry is the following:

\begin{conjecture}[Slicing Conjecture~\cite{B86}]
  \label{conj:slicing}
  There exists an absolute constant $C > 0$, such that $L_K \leq C$
  for all $n \geq 0$ and any convex body $K \subseteq \R^n$.
\end{conjecture}

The original bound computed by Bourgain~\cite{B86} was $L_K =
O(n^{1/4} \log n)$.  This has since been improved by
Klartag~\cite{K06} to $L_k = O(n^{1/4})$. In addition, the conjecture has been
verified for many classes of bodies including the $\ell_p$ norm balls.

The above concepts (centroid, covariance, isotropic constant, inertial
ellipsoid) all generalize easily to logconcave functions in lieu of
convex bodies; see Appendix~\ref{sec:additional-background} for
details.

\paragraph{Computational model.} All our algorithms will work with
convex bodies and norms presented by oracles in the standard way. The
complexity of our algorithms will be measured by the number of
arithmetic operations as well as the number of calls to the
oracle. See Appendix~\ref{sec:additional-background} for a more
detailed description of the kinds of oracles we use.


\section{Computing M-Ellipsoids and Coverings}
\label{sec:m-ellipsoid-covering}

An M-ellipsoid of a convex body $K$ is an ellipsoid $E$ with the
property that at most $2^{O(n)}$ translated copies of $E$ are
sufficient to cover all of $K$, and at most $2^{O(n)}$ copies of $K$
are sufficient to cover $E$.  More precisely, for any two subsets $A,B
\in \R^n$, define the covering number
\begin{equation}
  N(A,B) = \min \set{|\Lambda|: \Lambda \subseteq \R^n, A \subseteq B + \Lambda}.
  \label{def:cov-num}
\end{equation}
Hence $N(A,B)$ is the minimum number of translates of $B$ needed to
cover $A$.  The following theorem was first proved for symmetric
bodies by Milman~\cite{M86} and extended by Milman and
Pajor~\cite{MP00} to the general case.

\begin{theorem}[\cite{MP00}]
  \label{thm:m-ell-exist}
  There exists an absolute constant $C > 0$, such that for all $n \geq
  1$ and any convex body $K \subseteq \R^n$, there exists an ellipsoid
  $E$ satisfying
  \begin{equation}
    \label{eq:m-ell-exist}
    N(K,E) \cdot N(E,K) \leq C^n.
  \end{equation}
\end{theorem}

\begin{definition}[M-ellipsoid]
  Let $K \subseteq \R^n$ be a convex body.  If $E$ is an ellipsoid
  satisfying Equation~\eqref{eq:m-ell-exist} (for some particular
  fixed $C$) with respect to $K$, then we say that $E$ is an M-ellipsoid of $K$.
  \label{def:m-ell}
\end{definition}

\noindent There are many equivalent ways of understanding the
M-ellipsoid; here we list a few (proofs of many of these equivalences
can be found in~\cite{MP00}).
\begin{theorem} 
  \label{thm:m-ellipsoid-equiv}
  Let $K \subseteq \R^n$ be a convex body with $b(K) = 0$ (centroid at
  the origin), and let $E \subseteq \R^n$ be an origin-centered
  ellipsoid. Then the following conditions are equivalent, where the
  absolute constant $C$ may vary from line to line:
\begin{enumerate}[itemsep=0pt]
\item $N(K,E) \cdot N(E,K) \leq C^n$.
\item $\vol(K+E) \leq C^n \cdot \min \set{ \vol(E), \vol(K) }$.
\item $\sup_{t \in \R^n} \vol(K \cap (t+E)) \geq C^{-n} \cdot \max
  \set{\vol(E), \vol(K)}$. \label{item:m-ellipsoid-volume}
\item $E^{*}$ is an M-ellipsoid of $K^{*}$.
\end{enumerate}
\end{theorem}

From the above we see that the M-ellipsoid is very robust object, and
in particular is stable under polarity (assuming $K$ is
well-centered).  We will use this fact (or a slight variant of it) in
what follows, to help us certify a candidate M-ellipsoid.

As mentioned in the introduction, an M-ellipsoid for an $\ell_p$ ball
is trivial to compute.  Using condition~\ref{item:m-ellipsoid-volume}
of Theorem~\ref{thm:m-ellipsoid-equiv} above and standard volume
estimates for $\ell_p$ balls, i.e., that $\vol(B_p^n)^{1/n} =
\Theta(n^{-1/p})$, we have the following:

\begin{lemma}
  \label{lem:m-ellipsoid-lp}
  Let $B_p^n$ denote the $n$-dimensional $\ell_p$ ball. Then
  \begin{itemize}[itemsep=0pt]
  \item For $1 \leq p \leq 2$, $n^{\nicefrac{1}{2}-\nicefrac{1}{p}}
    \cdot B_2^n \subseteq B_p^n$ (the largest inscribed ball in
    $B_{p}^{n}$) is an M-ellipsoid for $B_p^n$.
  \item For $p \geq 2$, $n^{\nicefrac{1}{2}-\nicefrac{1}{p}} \cdot
    B_2^n \supseteq B_{p}^{n}$ (the smallest containing ball of
    $B_{p}^{n}$) is an M-ellipsoid for $B_p^n$.
  \end{itemize}
\end{lemma}

For general convex bodies, the proofs of existence of an M-ellipsoid
in \cite{M86} and \cite{MP00} are non-constructive.  It is worth
noting, however, that under the {\em slicing conjecture} (also known
as the {\em hyperplane conjecture}), a $\sqrt{n}$ scaling of $K$'s
inertial ellipsoid is an M-ellipsoid --- indeed, this is an equivalent
form of the slicing conjecture.  For many norms, including $\ell_p$,
absolutely symmetric norms (where the norm is preserved under
coordinate sign flips), and other classes, the slicing conjecture has
been proved.  Therefore, for such norms, an M-ellipsoid computation is
straightforward: using random walk techniques, estimate the covariance
matrix $\cov(K)$ of $K$, the unit ball of the norm, and return a
$\sqrt{n}$ scaling of $K$'s inertial ellipsoid.

In the rest of this section, we describe how to generate an
M-ellipsoid in general, without directly relying on the slicing
conjecture, with good probability in probabilistic polynomial time.
Moreover, we show how to certify that an ellipsoid is an M-ellipsoid
in deterministic $2^{O(n)}$ time.  A by-product of the certification
is a covering of the target body by at most $2^{O(n)}$ translates of
the candidate M-ellipsoid.  Such a covering will be used by all the
lattice algorithms in this paper.

Proofs for all the theorems in this section can be found in
Appendix~\ref{sec:m-ellipsoid-proofs}.

\subsection{The Main Algorithm}
\label{sec:main-algorithm}

The main result of this section is Algorithm~\ref{alg:M-Ellipsoid}
(M-Ellipsoid), whose correctness is proved in
Theorem~\ref{thm:m-ellipsoid}.  The algorithm uses two main
subroutines.  The first, M-Gen, described in
Section~\ref{sec:generate-candidate-m} below, produces a candidate
ellipsoid that is an M-ellipsoid with good probability.  The second,
Build-Cover, described in Section~\ref{sec:building-covering}, is used
to check that both $N(K,E), N((K-K)^*,E^{*}) = 2^{O(n)}$ by
constructing explicit coverings (if possible).  Because $N(E,K)
\approx N((K-K)^{*},E^{*})$ (up to $2^{\Theta(n)}$ factors) by the
duality of entropy (Theorem~\ref{thm:dual-entr}), such coverings
suffice to prove that $E$ is an M-ellipsoid for $K$.



\begin{algorithm}
  \caption{M-Ellipsoid: Generate a guaranteed M-ellipsoid and its
    implied covering.}
  \label{alg:M-Ellipsoid}
  \begin{algorithmic}[1]
    \REQUIRE A weak membership oracle $O_K$ for a $(0,r,R)$-centered
    convex body $K$.
    
    \ENSURE An M-ellipsoid $E$ of $K$, and a covering of $K$ by
    $2^{O(n)}$ copies of $E$.

    \STATE Approximate the centroid of $K$ using algorithm
    Estimate-Centroid (Lemma \ref{lem:estimate-centroid}).  If
    Estimate-Centroid fails, restart; otherwise, let $b$ denote
    returned estimate for $b(K)$.

    \STATE Generate a candidate M-ellipsoid $E$ of $K$ using
    Algorithm~\ref{alg:M-Gen} (M-Gen) on $K-b$.

    \STATE Check if $N(K,E) > (13e)^n$ using
    Algorithm~\ref{alg:build-cover} (Build-Cover).  If yes, restart;
    otherwise, let $T$ denote the returned covering of $K$ by $E$.

    \STATE Check if $N((K-K)^*,E^*) > (25 e \cdot 13)^n$ using
    Algorithm~\ref{alg:build-cover} (Build-Cover).  If yes, restart;
    otherwise, return $(E,T)$.
  \end{algorithmic}
\end{algorithm}

\begin{theorem}[Correctness of M-Ellipsoid]
  \label{thm:m-ellipsoid}
  For large enough $n$, Algorithm~\ref{alg:M-Ellipsoid} (M-Ellipsoid)
  outputs an ellipsoid $E$ satisfying
  \begin{equation}
    N(K,E) \leq \left(\sqrt{8\pi e} \cdot 13 e\right)^n
    \quad \text{and} \quad N(E,K) \leq \left(\sqrt{8\pi e} \cdot 25e 
      \cdot 13 \cdot 289 \right)^n
  \end{equation}
  along with a set $T \subseteq \Q^n$, $|T| \leq \left(\sqrt{8\pi e}
    \cdot 13 e\right)^n$ such that $K \subseteq T + E$, in expected
  time $\left(\sqrt{8 \pi e} \cdot 25 e \cdot 13 \right)^n \cdot
  \poly(n, \log(\tfrac{R}{r}))$.
\end{theorem}

\subsection{Generating a Candidate M-Ellipsoid}
\label{sec:generate-candidate-m}

Our algorithm for generating a candidate M-ellipsoid is based on a
constructive proof of Theorem~\ref{thm:m-ell-exist} by
Klartag~\cite{K06}, who suggested to us the idea of using these
techniques to build an M-ellipsoid algorithmically.  The main theorem
of~\cite{K06}, reproduced below, does not explicitly refer to
M-ellipsoids; instead, it shows that for every convex body $K$, there
is another convex body $K'$ that sandwiches $K$ between two small
scalings and satisfies the slicing conjecture.

\begin{theorem}[\cite{K06}]
  \label{thm:eps-slicing}
  Let $K \subseteq \R^n$ be a convex body. Then for every real $\eps
  \in (0,1)$, there exists a convex body $K' \subseteq \R^n$ such that
  \begin{equation}
    d(K,K') = \inf \set{\frac{b}{a}: \exists ~t \in \R^n \st aK' \subseteq K-t \subseteq bK'} \leq 1+\eps \quad \text{ and
    } \quad L_{K'} \leq \frac{c}{\sqrt{\eps}}.
  \end{equation}
  where $c > 0$ is an absolute constant and $L_{K'}$ is the isotropic
  constant of $K'$.
\end{theorem}

From the closeness of $K$ and $K'$ it follows that an M-ellipsoid for
$K'$ is an M-ellipsoid for $K$, and from the bound on $L_{K'}$ the
inertial ellipsoid of $K'$ is an M-ellipsoid for $K'$.

Here we will not need to construct $K'$ itself, but only an ellipsoid
very close to its inertial ellipsoid (which as just mentioned is an
M-ellipsoid for $K$).  The body $K'$ is derived from a certain family
of reweighted densities over $K$.  These densities are given by
exponential reweightings of the uniform density along some vector $s
\in \R^{n}$, i.e., $f_{s}(x) = e^{\inner{s,x}}$ for $x \in K$ (and $0$
otherwise).  For $s$ chosen uniformly from $n \cdot \conv\set{K-b(K),
  b(K)-K}^{*}$, the reweighting $f_{s}$ has two important properties:
(i) it is not too highly biased away from uniform over $K$, and (ii)
it has bounded isotropic constant (independent of $n$) with very high
probability.  Let $E$ be the inertial ellipsoid of $f_{s}$ (or any
reasonably good approximation to it), which can be found by sampling
from $f_{s}$.  The first property of $f_{s}$ allows us to prove that
$E$ can be covered by $2^{O(n)}$ copies of $K$, while the second
property lets us cover $K$ by $2^{O(n)}$ copies of $E$ (see
Lemma~\ref{lem:iner-to-m}).

To make everything work algorithmically, we need robust versions of
Klartag's main lemmas, since we will only be able to compute an
approximate centroid of $K$, sample $s$ from a distribution close to
uniform, and estimate the covariance matrix of $f_{s}$.

Algorithm~\ref{alg:M-Gen} makes the above description more formal.
Note that given an oracle for a convex body, an oracle for the polar
body can be constructed in polynomial time~\cite{GLS}.  Sampling, both
from the uniform and exponentially reweighted distributions, can be
done in polynomial time using the random walk algorithm
of~\cite{LV3,LV06}.  Theorem~\ref{lem:iner-to-m} together with
Lemma~\ref{lem:exp-slice} implies that the algorithm's output is
indeed an M-ellipsoid with good probability.

\begin{algorithm}
  \caption{M-Gen: Randomized generation of a candidate M-ellipsoid.}
  \label{alg:M-Gen}
  \begin{algorithmic}[1]
    \REQUIRE A weak membership oracle $O_K$ for a $(0,r,R)$-centered
    convex body $K$ with $b(K) \in \frac{1}{n+1}E_K$.
    
    \ENSURE With probability $1-o(1)$, an M-ellipsoid of $K$.

    \STATE Estimate the centroid $b=b(K)$ using uniform samples from
    $K$.

    \STATE Construct a membership oracle for $n\left(\conv
      \set{K-b,b-K}\right)^{*}$.

    \STATE Sample a random vector $s$ from $n\left(\conv
      \set{K-b,b-K}\right)^*$.

    \STATE Estimate the covariance matrix $A$ of the density
    proportional to $e^{\inner{s,x}}$, restricted to $K$.

    \STATE Output the ellipsoid $E(A^{-1}) = \set{x : x^tA^{-1}x \leq
      1}$.
  \end{algorithmic}
\end{algorithm}

\begin{theorem}[Correctness of M-Gen]
  \label{thm:m-gen}
  For large enough $n$, Algorithm~\ref{alg:M-Gen} (M-Gen) outputs an
  ellipsoid $E$ satisfying
  \begin{equation}
    N(E,K) \leq (25e)^n \quad \text{and} \quad N(K,E) \leq (13e)^n
  \end{equation}
  with probability at least $1-\frac{3}{n}$ in time $\poly(n,
  \log(\tfrac{R}{r}))$.
\end{theorem}

\subsection{Building a Covering}
\label{sec:building-covering}

The next theorem yields an algorithm to approximately decide (up to
single exponential factors) whether a given convex body $K$ can be
covered by a specified number of translates of an ellipsoid $E$. The
algorithm is constructive and proceeds by constructing a simple
parallelepiped tiling of $K$, where the parallelepiped in question is
a maximum volume inscribed parallelepiped of $E$.

\begin{algorithm}
  \caption{Build-Cover: Deterministic construction of an ellipsoid
    covering of a convex body.}
  \label{alg:build-cover}
  \begin{algorithmic}[1]
    \REQUIRE A weak membership oracle $O_K$ for an $(0,r,R)$-centered
    convex body $K$, an ellipsoid $E=E(A)$, and some $H \geq 1$.

    \ENSURE Either a covering of $K$ by $(\sqrt{8 \pi e}H)^{n}$
    translates of $E$, or a declaration that $K$ cannot be covered by
    $H^{n}$ copies of $E$.

    \STATE Let $C_E$ be any maximum-volume inscribed parallelepiped of
    $E$ (e.g., a maximum-volume inscribed cuboid with the same axes as
    the ellipsoid).

    \STATE Attempt to cover $K$ using translates of $C_E$ with respect
    to the natural parallelepiped tiling, via a breadth-first search
    over the tiling lattice, starting from the origin.

    \STATE If the attempted covering grows larger than $(\sqrt{8 \pi
      e}H)^{n}$, abort.  Otherwise, output the covering.
  \end{algorithmic}
\end{algorithm}

\begin{theorem}
  \label{thm:build-cover}
  Algorithm~\ref{alg:build-cover} (Build-Cover) is correct, and runs
  in time $\left(\sqrt{8\pi e} H\right)^n \cdot
  \poly(n,\enc{A}, \log(\tfrac{R}{r}))$.
\end{theorem}


\section{Lattice Algorithms}
\label{sec:lattice-algs}

In this section we prove our general enumeration theorem for convex
bodies (Theorem~\ref{thm:body-lat-enum-informal}, formalized in
Theorem~\ref{thm:body-lat-enum}) and give its application to the
Shortest and Closest Vector Problems, and Integer Programming.

\subsection{Lattice Background}

An $n$-dimensional lattice $L \subset \R^{n}$ is a discrete subgroup
under addition.  It can be written as
\begin{equation}
  L = \setfit{\sum_{i=1}^k z_i b_i : z_i \in \Z}
\end{equation}
for some (not necessarily unique) \emph{basis} $B=(b_1,\dots,b_k)$ of
$k \leq n$ linearly independent vectors in $\R^n$.  The determinant of
$L$ is defined as
\begin{equation}
  \det(L) = \sqrt{\det(B^t B)}.
\end{equation}
The
\emph{dual lattice} $L^*$ of $L$ is defined as
\begin{equation}
  L^* = \set{y \in \mathrm{span}(b_1,\dots,b_k) : \forall x \in L, \inner{x,y} \in \Z}.
\end{equation}

The \emph{minimum distance} of $L$ with respect to $K$ is
$\lambda_{1}(K, L) = \min_{y \in L \setminus \set{0}} \length{y}_{K}$.
The \emph{covering radius} of $L$ with respect to $K$ is $\mu(K,L) =
\inf \set{s \geq 0: L + sK = \R^n}$.  Note that from the definition,
we see that $\mu(K+t,L) = \mu(K,L)$ for $t \in \R^n$ and that
$\mu(-K,L) = \mu(K,L)$.  We also define $d_K(L,x) = \inf_{y \in L}
\length{y-x}_K$. We define the \emph{$i^{th}$ minimum} of $L$ with respect to the $\ell_2$ norm as
\[
\lambda_i(L) = \inf \set{r \geq 0: \dim(\mathrm{span}(rB_2^n \cap L)) \geq i}
\]
where $\mathrm{span}$ denotes the linear span.

The \emph{shortest vector problem} (SVP) with respect to $K$
is the following: given a basis of an $n$-dimensional lattice $L$,
compute an element of
\begin{equation}
  \SVP(K,L) = \argmin_{y \in L \setminus \set{0}}\; \length{y}_K.
\end{equation}
The \emph{closest vector problem} (CVP) with respect to $K$ is: given
a basis of an $n$-dimensional lattice $L$ and a point $x \in \R^n$, compute an
element of
\begin{equation}
  \CVP(K,L,x) = \argmin_{y \in L} \; \length{y-x}_K.
\end{equation}


To denote the sets of \emph{approximate} minimizers for SVP and CVP, we
define for any $\epsilon > 0$
\begin{align}
  \SVP_{\epsilon}(K,L) &= \set{z \in L \setminus \set{0} :
    \length{z}_{K} \leq (1+\epsilon) \cdot \min_{y \in L \setminus \set{0}}
    \length{y}_{K}} \\
  \CVP_{\epsilon}(K,L,x) & = \set{z \in L:
    \length{z-x}_K \leq (1+\eps) \cdot \min_{y \in L} \length{y-x}_K}.
\end{align}

\paragraph{Integer programming.}

A fundamental tool in integer programming is
the so-called ``flatness theorem,'' which says that for any convex
body $K \subseteq \R^n$ and $n$-dimensional lattice $L \subseteq
\R^n$,
\begin{equation}
\label{eq:flatness}
1 \leq \mu(K,L) \cdot \lambda_1((K-K)^*,L^*) \leq f(n),
\end{equation}
where $\mu(K,L) = \inf \set{s \geq 0: L + sK = \R^n}$ is the covering
radius of $L$, and
\[
\lambda_1((K-K)^*,L^*) = \inf_{y \in L^* \setminus \set{0}}
\big(\sup_{x \in K} \pr{x}{y} - \inf_{x \in K} \pr{x}{y}\bigr)
\]
is the lattice width of $K$.  The flatness theorem is most easily
interpreted as follows: either $K$ certainly contains a lattice point
in $L$, or there exist at most $\floor{f(n)}+1$ hyperplanes of the
form $H_{k} = \set{x \in \R^n: \pr{y}{x} = k}$, $y \in L^* \setminus
\set{0}$, $k \in \Z$ and $\inf_{x \in K} \pr{y}{x} \leq k \leq \sup_{x
  \in K} \pr{y}{x}$, such that any lattice point in $K$ must lie on
one of these hyperplanes.  Crucially, we note that computing
$\lambda_1((K-K)^*,L^*)$ for a general convex body $K$ is exactly a
shortest non-zero vector computation with respect to a general norm.

The asymptotic growth (and even the finiteness) of the function $f(n)$
in~\eqref{eq:flatness} has been the source of intense study over the
past century. Restricting to the important special case where
$K=B_2^n$, the optimal growth rate has been settled at $f(n) =
\Theta(n)$~\cite{banaszczyk93:_new}. When $K$ is centrally symmetric,
the best known bound is $f(n) = O(n \log n)$~\cite{Bana96}.  For the
general case, the current best bound is $f(n) =
O(n^{\nicefrac{4}{3}}\log^c n)$ \cite{Bana99,R00} for some fixed $c >
0$.  We let $f^*(n)$ denote best possible upper bound for the general
flatness theorem.

\subsection{Lattice Point Enumeration in Convex Bodies}

We now use enumeration via the M-ellipsoid covering to solve the
Shortest and Closest Vector Problems. To do this we will need the
recent algorithm of Micciancio and
Voulgaris~\cite{DBLP:conf/stoc/MicciancioV10} for the Closest Vector
Problem under the $\ell_2$ norm (and hence any ellipsoidal norm),
which we call the MV algorithm for short.  The following is an
immediate extension of their graph-traversal approach~\cite{VD10}.

\begin{proposition}[\cite{DBLP:conf/stoc/MicciancioV10}, Algorithm
  Ellipsoid-Enum]
  There is an algorithm Ellipsoid-Enum that, given any positive
  definite $A \in \Q^{n \times n}$, any basis $B$ of an
  $n$-dimensional lattice $L \subseteq \R^{n}$, and any $t \in \R^n$,
  computes the set $L \cap (E(A)+t)$ in deterministic time
  \begin{equation}
    2^{O(n)} \cdot (|L \cap (E(A)+t)|+1) \cdot \poly(\enc{A},\enc{B},\enc{t}).
  \end{equation}
\end{proposition}
Here the idea is that the points inside $(E(A)+t) \cap L$ form a
connected subgraph, where we consider two lattice points adjacent if
they differ by a Voronoi-relevant vector of $L$, where Voronoi
relevance is defined with respect to the inner product defined by $A$
(see \cite{DBLP:conf/stoc/MicciancioV10} for formal definitions).  An
initial point inside $(E(A) + t) \cap L$ can be computed (if it
exists) in a single call to the MV algorithm, and the rest can be
computed by a standard breadth-first search of the graph.

For a convex body $K \subseteq \R^n$ and a lattice $L \subseteq \R^n$ define
\begin{equation}
  G(K,L) = \max_{x \in \R^n} |(K+x) \cap L|,
\end{equation}
the maximum number of lattice points in $K$ under any translation.

We can now state our enumeration theorem, which formalizes Theorem~\ref{thm:body-lat-enum-informal} from the
introduction.

\begin{algorithm}
  \caption{Algorithm Lattice-Enum$(K, L, x, d, \eps)$}
  \label{alg:lattice-enum-algorithm}
  \begin{algorithmic}[1]

    \REQUIRE An $(0,r,R)$-centered convex body $K$ presented by a weak
    distance oracle $D_K$ for $\length{\cdot}_K$, a basis $B$ for a
    lattice $L$, an input point $x$, distance $d \geq 0$, and $0 <
    \eps < 1$.

    \ENSURE $S \subseteq L$ satisfying $(\ref{eq:ble-1})$.

    \STATE Let $(E,T) \leftarrow \text{M-Ellipsoid}(K)$ \COMMENT{This
      covering need only be computed once for repeated calls.}

    \STATE Let $S \leftarrow \emptyset$

    \FORALL{$s \in T$}
    
    \STATE Let $U_s \gets \text{Ellipsoid-Enum}(dE,~L, x+ds)$
    
    \STATE $S \leftarrow S \cup \set{y: y \in U_S, D_K(y-x,
      \nicefrac{\eps}{2}) \leq d + \nicefrac{\eps}{2}}$
    
    \ENDFOR

    \RETURN $S$
  \end{algorithmic}
\end{algorithm}

\begin{theorem}[Enumeration in convex bodies]
  \label{thm:body-lat-enum}
  Algorithm~\ref{alg:lattice-enum-algorithm} (Lattice-Enum) outputs a
  set $S \subseteq L$ such that
  \begin{equation}
    \set{y \in L: \length{y-x}_K \leq d} \subseteq S \subseteq \set{y \in
      L: \length{y-x}_K \leq d + \eps}
    \label{eq:ble-1}
  \end{equation}
  in expected time $G(dK,L) \cdot 2^{O(n)} \cdot
  \poly(\log(\nicefrac{R}{r}), \log(\nicefrac{1}{\epsilon}), \enc{B},
  \enc{x})$.
\end{theorem}

\begin{proof} \ \vspace{-12pt}
\paragraph{Correctness:} 
We first note that $K \subseteq \cup_{s \in T}~ s + E$ then $x+dK
\subseteq \cup_{s \in T}~ x + d(s + E)$. Hence given a covering for
$K$, we have a covering of $dK + t$. Now on input $(dE,~L,~x+ds)$ the
algorithm Ellipsoid-Enum returns the set $(x + ds) + dE \cap L = x +
d(s + E) \cap L$.

Now we first show that for all $y \in x + dK \cap L$, $y \in S$. By
the covering property, we know that for some $s \in T$, $y \in x +
(s+E) \cap L$. Finally, by the properties of the weak-semi norm oracle
since $y \in dK + x \Leftrightarrow \length{y-t}_K \leq d$, we have
that
\[
D_K(y-x,\nicefrac{\eps}{2}) \leq \length{y-x}_K + \nicefrac{\eps}{2}
\leq d + \nicefrac{\eps}{2} \text{,}
\]
and hence $y$ is correctly placed in $S$ as needed. Lastly, we must
show that if $y \notin (d+\eps)K + x \Leftrightarrow \length{y-t} >
d+\eps$, then $y \notin S$. Again, from the properties of the
weak distance oracle we see that
\[
D_K(y-x,\nicefrac{\eps}{2}) \geq \length{y-x}_K - \nicefrac{\eps}{2} >
d + \eps - \nicefrac{\eps}{2} = d + \nicefrac{\eps}{2}
\]
as needed. Lastly, by construction, the set $S$ only contains lattice
points, and so by the above arguments $U$ satisfies the required
properties.

\paragraph{Runtime:} By Theorem~\ref{thm:m-ellipsoid}, M-Ellipsoid computes
an M-ellipsoid in expected time $\polylog(\nicefrac{R}{r})C_1^n$. Let
$E$ denote an M-ellipsoid of $K$ and let $T \subseteq \R^n$ be as
above.  From Theorem \ref{thm:m-ellipsoid}, we know that $|T| \leq C_2^n$,
hence the algorithm makes at most $C_2^n$ calls to Ellipsoid-Enum. Now
to bound the complexity of enumerating $x+d(s+E) \cap L$ for each $s
\in T$, we need to bound $|x+d(s+E) \cap L| \leq G(dE,L)$. Now we note
that
\[
G(dE,L) \leq N(dE,dK)G(dK,L) = N(E,K)G(dK,L)\leq C_2^n G(dK,L)
\]
by Theorem \ref{thm:m-ellipsoid}. Hence for any $s \in T$, Ellipsoid-Enum
takes at most $C_3^n ~\poly(\enc{B}, \enc{x}) ~ (C_2^n G(dK,L)) \leq
\poly(\enc{B}, \enc{x}) ~ C_4^n ~ G(dK,L)$ time to compute $x + d(s+E)
\cap L$. Hence the total running time is bounded by
\begin{equation}
  \polylog(\tfrac{R}{r}) ~ C_1^n + C_2^n ~ \poly(\enc{B}) ~ C_4^n ~ G(dK,L) ~\leq~ 
  \poly(\log \tfrac{R}{r},\enc{B},\enc{x}) ~ C_5^n ~ G(dK,L)
\end{equation}
where $C_5 > 0$ is an absolute constant.
\end{proof}

We remark that the only randomness in the algorithm is to build the
M-ellipsoid; once this has been achieved the rest of the algorithm is
deterministic. Hence, in the cases where the M-ellipsoid is known
explicitly, as it is for the $\ell_p$ balls (where an appropriately
scaled Euclidean ball suffices), the algorithm can be in fact
made completely deterministic. The algorithms for the shortest vector
and closest vector problem described in the next sections will only
depend on the Lattice-Enum algorithm, and hence they will be
deterministic as long as Lattice-Enum is deterministic.

\subsection{Shortest Vector Problem}
\label{sec:short-vect-probl}

Our main goal will be to use the above enumeration algorithm to solve the Shortest Vector Problem. The following gives a
useful bound on $G(K,L)$ for a general convex body.

\begin{lemma}
  \label{lem:lambda1-bd}
Let $K \subseteq \R^n$ be a convex body satisfying $\vol(K \cap -K) \geq \gamma^{-n} \vol(K)$, $\gamma \geq
1$, and let $L$ be an $n$-dimensional lattice. Then for $d > 0$ we have that
  \begin{equation}
    G(dK,L) \leq \left(\gamma\left(1 + \frac{2d}{\lambda_1(K,L)}\right)\right)^n.
  \end{equation}
\end{lemma}

We note $\gamma$ above is easily bounded in many natural situations. When $K$ is centrally symmetric we can set $\gamma
= 1$ since $K \cap -K = K$, and if $K$ is a general convex body with $b(K) = 0$ setting $\gamma=2$ is valid by
Theorem~\ref{thm:symmetrize-ext}. Hence the notion of ``well-centered'', i.e., $\gamma \leq 4$, is quite robust.

\begin{proof}[Proof of Lemma~\ref{lem:lambda1-bd}]
  Let $s = \nicefrac{1}{2} \lambda_1(K,L)$. For $x \in L$, we examine
  \[ x + \mathrm{int}(s(K \cap -K)) = \set{z \in \R^n: \length{z-x}_{K
      \cap -K} < s}.\] Now for $x, y \in L$, $x \neq y$, we claim that
  \begin{equation}
    x + \mathrm{int}(s(K \cap -K)) \cap y + \mathrm{int}(s(K \cap -K)) = \emptyset
    \label{eq:no-int-1}
  \end{equation}
  Assume not, then $\exists ~ z \in \R^n$ such that
  $\length{z-x}_{ K\cap-K },\length{z-y}_{ K \cap -K} < s$. Since $K \cap -K$ is symmetric, we note
  that $\length{y-z}_{K \cap -K}  = \length{z-y}_{K \cap -K} < s$. But now, since $K \cap -K \subseteq K$, we see that
  \begin{align*}
    \length{y-x}_K &= \length{y-z+z-x}_K \leq \length{y-z}_K + \length{z-x}_K \\
    &\leq \length{y-z}_{K \cap -K} + \length{z-x}_{K \cap -K} < s + s = 2s = \lambda_1(K,L)
  \end{align*}
  a clear contradiction since $y-x \neq 0$.

  Take $c \in \R^n$. To bound $G(dK,L)$ we must bound $|(c + dK) \cap L|$. For $x \in c + dK$, we note that $x + s(K \cap
  -K) \subseteq c + (d+s)K$. Therefore,
  \begin{equation}
    \vol((d+s) K) = \vol(c + (d+s)K) \geq \vol\left( ((c + dK) \cap L) + s(K \cap -K) \right) = |(c + dK) \cap L|
    \vol(s(K \cap -K))
  \end{equation}
  where the last equality follows from $(\ref{eq:no-int-1})$. Therefore, we have that
  \begin{equation}
    |(c + dK) \cap L| \leq \frac{\vol((d+s)K)}{\vol(s(K \cap -K))} = \left(\frac{d+s}{\gamma^{-1}s}\right)^n 
    = \left(\gamma\left(1 + \frac{2d}{\lambda_1(K,L)}\right)\right)^n
  \end{equation}
  as needed.
\end{proof}

We can now state the algorithm and main theorem of this section.

\begin{algorithm}
  \caption{Shortest-Vectors$(K, L,\eps)$}
  \label{alg:svp-algorithm}
  \begin{algorithmic}[1]

    \REQUIRE A $(0,r,R)$-centered convex body $K$ presented by a weak
    distance oracle $D_K$ for $\length{\cdot}_K$, a basis $B$ for a
    lattice $L$, and $0 < \eps < 1$.

    \ENSURE $S \subseteq L$ such that $\SVP(K,L) \subseteq S \subseteq
    \SVP_{\eps}(K,L)$

    \STATE Compute $z \in \SVP(B_2^n, L)$ using the MV algorithm. Set
    $t,d \leftarrow \nicefrac{\length{z}}{R}$.

    \REPEAT

    \STATE $U \leftarrow \text{Lattice-Enum}(K, L,0,d,t)
    \setminus \set{0}$

    \IF{$U = \emptyset$}

    \STATE $d \leftarrow 2 d$
    
    \ENDIF
    
    \UNTIL{$U \neq \emptyset $}

    \STATE $U \leftarrow \text{Lattice-Enum}(K, L ,0,d+t,t)
    \setminus \set{0}$

    \STATE $m \leftarrow \min \set{D_K(y, \frac{\eps}{4}t): y \in
      U}$
    
    \STATE $S \leftarrow \set{y: D_K(y, \frac{\eps}{4}t) \leq m +
      \frac{\eps}{2} ~t, y \in U}$

    \RETURN $S$
  \end{algorithmic}
\end{algorithm}

\begin{theorem}[Correctness of Shortest-Vectors]
  \label{thm:svp-alg}
  If $K$ is well-centered, i.e., $\vol(K \cap -K) \geq 4^{-n}
  \vol(K)$, then Algorithm~\ref{alg:svp-algorithm} (Shortest-Vectors)
  outputs a set $S \subseteq L$ satisfying
  \begin{equation}
    \SVP(K,L) \subseteq S \subseteq \SVP_{\epsilon}(K,L)
  \end{equation} 
  in expected time
  \begin{equation}
    2^{O(n)} \cdot \poly(\log(\tfrac{R}{r}), \log(\tfrac{1}{\eps}), \enc{B}).
  \end{equation}
\end{theorem}

\begin{proof} \ \vspace{-12pt}  
  \paragraph{Correctness:} First note that since $K$ is
  $(0,r,R)$-centered, we know that $\frac{\length{y}}{R} \leq
  \length{y}_K \leq \frac{\length{y}}{r}$ for all $y \in \R^n$. Now
  take $z \in \SVP(K, L)$ and $z' \in \SVP(B_2^n, L)$. Let $\omega =
  \length{z}_K$, and as in the algorithm let $t =
  \frac{\length{z'}}{R}$. Now we have that
  \begin{equation}
    t = \frac{\length{z'}}{R} \leq \frac{\length{z}}{R} \leq \length{z}_K \leq \length{z'}_K \leq \frac{\length{z'}}{r} = t \frac{R}{r}
  \end{equation}
  Therefore $t \leq \omega \leq t \frac{R}{r}$.

  Now for $z \in \SVP(K,L)$, we must show that $z \in S$. Let $d_f$
  denote the final value of $d$ after the while loop terminates. Since
  $U \neq \emptyset$ and $0 \notin U$ after the while loop terminates,
  and since the enumeration algorithm guarantees that $U \subseteq
  \set{y \in L: \length{y-x}_K \leq d_f+t}$, we have that $\omega \leq
  d_f+t$.  Now let $U_f = \text{Lattice-Enum}(K,L,0,d_f+t,t)
  \setminus \set{0}$, i.e. the final setting of the set $U$. By the
  properties of Lattice-Enum, we know that $\set{y \in L:
    \length{y-x}_K \leq d_f+t} \subseteq U_f$, and hence we have that
  $\SVP(K,L) \subseteq U_f$. From the computation of the number $m$,
  during the final stage of the algorithm, we now see that $\omega -
  \frac{\eps}{4}t \leq m \leq \omega + \frac{\eps}{4}t$. Therefore for
  $z \in \SVP(K,L)$, we have that
  \begin{equation}
    D_K(z, \frac{\eps}{4}t) \leq \omega + \frac{\eps}{4}t \leq m + \frac{\eps}{2}t
  \end{equation}
  and hence $z$ will correctly be placed in $S$ as needed.

  Now assume that $z \in L \setminus \set{0}$ and $z \notin
  \SVP(K,L)_\eps$. We must show that $z \notin S$. Since $\omega \geq
  t$ from above, we have that $\length{z}_K > (1+\eps)\omega \geq
  \omega + \eps t$. Therefore, we see that
  \begin{equation}
    D_K(z, \frac{\eps}{4}t) \geq \length{z}_K - \frac{\eps}{4}t > \omega + \frac{3\eps}{4}t \geq m + \frac{\eps}{2}t
  \end{equation}
  and hence $z$ will never be added to $S$ as needed.

  \paragraph{Runtime:} First we run MV to compute an element of
  $\SVP(B_2^n,L)$ which takes $\poly(\enc{B},\enc{x})2^{O(n)}$
  time. Next since $\omega \geq t$ ($\omega,t$ as above), we have that
  $\lambda_1(K,L) \geq t$. Now the enumeration algorithm is seeded
  with $d = t \leq \lambda_1(K,L)$. From here we see that the moment
  $d$ is pushed above $\lambda_1(K,L)$, the set $U$ returned by
  Lattice-Enum will be non-empty. Hence during the execution of the
  while loop, the value of $d$ is never more that $2
  \lambda_1(K,L)$. Furthermore, the last execution of the enumeration
  algorithm is run on $d + t \leq 3\lambda_1(K,L)$. Hence every run of
  the enumeration algorithm happens for distances less than
  $3\lambda_1(K,L)$. Therefore by Lemma $\ref{lem:lambda1-bd}$ and
  Theorem $\ref{thm:body-lat-enum-informal}$, we have that each run of
  the enumeration algorithm takes at most
  \begin{equation}
    \polylog(\tfrac{R}{r}, \tfrac{1}{t}) ~ \poly(\enc{B})~ C^n ~ G(3\lambda_1(K,L)K,L) 
    \leq \polylog(\tfrac{R}{r}, \tfrac{1}{t}) ~ \poly(\enc{B})~ C^n ~ (4 \cdot 7)^n
    \label{eq:svp-alg-1}
  \end{equation}
  Next, since $t \leq \omega \leq t \frac{R}{r}$, we see that we will
  execute the enumeration algorithm at most $\log_2 \frac{R}{r} + 1$
  times. Remembering that $t = \frac{\length{z'}}{R}$, we have that
  all the lattice points of $L$ generated by the algorithm lie inside
  a ball of radius at most $3~\frac{R}{r}~\length{z'} \leq
  3~\frac{R}{r}~\sqrt{n}\length{B}$ around $x$.  Hence, these lattices
  points as well as the number $t$ can be represented using at most
  $\poly(\enc{B}, \enc{x}, \ln(\frac{R}{r}))$ bits. Therefore, apart
  from in the enumeration algorithm, we only evaluate the weak norm
  oracle on inputs of size $\poly(\enc{B}, \ln(\frac{R}{r}), \enc{x},
  \ln \frac{1}{\eps})$ which is polynomial in the input.  Finally, we
  filter the list $U_f$ into $S$, which requires exactly $2|U_f|$
  evaluations of the norm-oracle, where the cardinality of $U_f$ is
  bounded by (\ref{eq:svp-alg-1}). Combining all of the above bounds,
  yields the desired result.
\end{proof}

\subsection{Closest Vector Problem}
\label{sec:clos-vect-probl}

Before presenting our CVP algorithm, we again need a simple
enumeration bound.

\begin{lemma}
  \label{lem:gkl-smooth}
  Let $K \subseteq \R^n$ be a convex body, and let $L \subseteq \R^n$
  denote an $n$-dimensional lattice.  Then for $t > 0$ we have
  \begin{equation}
    G(tK,L) \leq (4t+2)^n \cdot G(K,L)
  \end{equation}
\end{lemma}

\begin{proof}
  Since $G(tK,L)$ is invariant under shifts of $K$, we may assume that
  $b(K) = 0$. Since $b(K) = 0$, from~\cite{MP00} we know that
  $\vol(K) \leq 2^n \vol(K \cap -K)$ (Theorem
  \eqref{thm:symmetrize}). We remember that $N(tK, K \cap -K)$ denotes
  the minimum number of translates of $K \cap -K$ needed to cover
  $tK$. Since $K \cap -K$ is symmetric, by a standard packing argument
  we have that
  \begin{align}
    \begin{split}
      N(tK, K \cap -K) &\leq \frac{\vol(tK + \frac{1}{2}(K \cap
        -K))}{\vol(\frac{1}{2}(K \cap -K))}
      \leq \frac{\vol(tK + \frac{1}{2}K)}{\vol(\frac{1}{2}(K \cap -K))} \\
      &= \left(\frac{t+\frac{1}{2}}{\frac{1}{2}}\right)^n
      \frac{\vol(K)}{\vol(K \cap -K)} \leq (2t+1)^n 2^n = (4t+2)^n.
    \end{split}
  \end{align}
  Next since $K \cap -K \subseteq K$, we have that $N(tK, K) \leq
  N(tK, K \cap -K)$. Now let $\Lambda \subseteq \R^n$ denote a set
  satisfying $|\Lambda| = N(tK, K)$ and $tK \subseteq \bigcup_{x \in
    \Lambda} x + K$. Then for $c \in \R^n$ we have that
  \begin{align}
    \begin{split}
      |tK + c \cap L| &\leq |(\Lambda + c + K) \cap
      L|
      \leq \sum_{x \in \Lambda} |(x + c + K) \cap L| \\
      &\leq |\Lambda| \cdot G(K,L) = N(tK, K) \cdot G(K,L) \leq
      (4t+2)^n \cdot G(K,L)
    \end{split}
  \end{align}
  as needed.
\end{proof}

We can now state the algorithm and main theorem of this section.

\begin{algorithm}
  \caption{Closest-Vectors$(K, L, x, \epsilon)$}
  \label{alg:cvp-algorithm}
  \begin{algorithmic}[1]

    \REQUIRE An $(0,r,R)$-centered convex body $K$ with weak distance oracle $D_K$ for
    $\length{\cdot}_K$, a basis $B$ for a lattice $L$, an input
    point $x$, and $0 < \eps < 1$.

    \ENSURE $S \subseteq L$, $\CVP(K,L,t) \subseteq S \subseteq
    \CVP_{\eps}(K,L,t)$

    \IF{$x \in L$}

    \RETURN $\set{x}$

    \ENDIF

    \STATE Compute $z \in \CVP(B_2^n, L)$ using the MV algorithm. Set
    $t,d \leftarrow \nicefrac{\length{z}}{R}$

    \REPEAT

    \STATE $U \leftarrow \text{Lattice-Enum}(K,L,x,d,t)$

    \IF{$U = \emptyset$}

    \STATE $d \leftarrow 2 d$

    \ENDIF

    \UNTIL{$U \neq \emptyset$}

    \STATE $U \leftarrow \text{Lattice-Enum}(K,L,x,d+t,t)$

    \STATE $m \leftarrow \min \set{D_K(y-x, \frac{\eps}{4}t): y \in
      U}$

    \STATE $S \leftarrow \set{y: D_K(y-x, \frac{\eps}{4}t) \leq m +
      \frac{\eps}{2} ~t, y \in U}$

    \RETURN $S$
  \end{algorithmic}
\end{algorithm}

\begin{theorem}[Correctness of Closest-Vectors]
  \label{thm:cvp-alg}
  If $K$ is well-centered, i.e., $\vol(K \cap -K) \geq 4^{-n}
  \vol(K)$, then Algorithm~\ref{alg:cvp-algorithm} computes a set $S
  \subseteq L$ such that
  \begin{equation}
    \CVP(K,L,x) \subseteq S \subseteq \CVP_{\epsilon}(K,L,x)
  \end{equation}
  in expected time
  \begin{equation}
    2^{O(n)} \cdot G(d K,L) \cdot \poly(\log(\tfrac{1}{\eps}),\log(\tfrac{R}{r}), \enc{B},\enc{x}),
  \end{equation}
  where $d = d_K(L,x)$.
\end{theorem}
The proof is essentially identical to the one for SVP.

\begin{proof}\ \vspace{-12pt}
  \paragraph{Correctness:} If $x \in L$, clearly there is nothing to do, so
assume $x \notin L$. First note that since $K$ is $(0,r,R)$-centered, we know
that $\frac{\length{y}}{R} \leq \length{y}_K \leq \frac{\length{y}}{r}$ for all
$y \in \R^n$. Now take $z \in \CVP(K, L, x)$ and $z' \in \CVP(B_2^n,L,x)$. Let
$\omega = \length{z-x}_K$, and as in the algorithm let $t =
\frac{\length{z'-x}}{R}$. Now we have that
  \begin{equation}
    t = \frac{\length{z'-x}}{R} \leq \frac{\length{z-x}}{R} \leq \length{z-x}_K \leq \length{z'-x}_K \leq \frac{\length{z'-x}}{r} = t \frac{R}{r}
  \end{equation}
  Therefore $t \leq \omega \leq t \frac{R}{r}$. Now for $z \in
  \CVP(K,L,x)$, we must show that $z \in S$. Let $d_f$ denote the
  final value of $d$ after the while loop terminates. Since $U \neq
  \emptyset$ after the while loop terminates, and since the
  enumeration algorithm guarantees that $U \subseteq \set{y \in L:
    \length{y-x}_K \leq d_f+t}$, we have that $\omega \leq d_f+t$.
  Now let $U_f = \mathrm{Enumerate}(K,L,x,d_f+t,t)$, i.e. the final
  setting of the set $U$.  By the properties Lattice-Enum, we know
  that $\set{y \in L: \length{y-x}_K \leq d_f+t} \subseteq U_f$, and
  hence we have that $\CVP(K,L,x) \subseteq U_f$. From the computation
  of the number $m$, during the final stage of the algorithm, we now
  see that $\omega - \frac{\eps}{4}t \leq m \leq \omega +
  \frac{\eps}{4}t$. Therefore for $z \in \CVP(K,L,x)$, we have that
  \begin{equation}
    D_K(z-x, \frac{\eps}{4}t) \leq \omega + \frac{\eps}{4}t \leq m + \frac{\eps}{2}t
  \end{equation}
  and hence $z$ will correctly be placed in $S$ as needed.

  Now assume that $z \in L$ and $z \notin \CVP(K,L,x)_\eps$. We must
  show that $z \notin S$. Since $\omega \geq t$ from above, we have
  that $\length{z-x}_K > (1+\eps)\omega \geq \omega + \eps
  t$. Therefore, we see that
  \begin{equation}
    D_K(z-x, \frac{\eps}{4}t) \geq \length{z-x}_K - \frac{\eps}{4}t > \omega + \frac{3\eps}{4}t \geq m + \frac{\eps}{2}t
  \end{equation}
  and hence $z$ will never be added to $S$ as needed.

  \paragraph{Runtime:} We first check if $x \in L$, this take
  $\poly(\enc{B},\enc{x})$ time. Next, we run the MV algorithm to
  compute an element of $\CVP(B_2^n,L,x)$ which takes
  $\poly(\enc{B},\enc{x})2^{O(n)}$ time. Next, note that since $\omega
  \geq t$ ($\omega,t$ as above), we have that $d_K(L,x) \geq t$. Now
  the enumeration algorithm is seeded with $d = t \leq d_K(L,x)$. Now
  we note that the moment $d$ is pushed above $d_K(L,x)$, the set $U$
  returned by the enumeration algorithm will be non-empty. Hence
  during the execution of the while loop, the value of $d$ is never
  more that $2 d_K(L,x)$. Furthermore, the last execution of the
  enumeration algorithm is run on $d + t \leq 3d_K(L,x)$.  Hence every
  run of the enumeration algorithm happens for distances less than
  $3d_K(L,x)$. Therefore by Lemma $\ref{lem:gkl-smooth}$ and Theorem
  $\ref{thm:body-lat-enum}$, we have that each run of the enumeration
  algorithm takes at most
  \begin{multline}
    \polylog(\nicefrac{R}{r}, \nicefrac{1}{t}) ~ \poly(\enc{B})~ C^n ~ G(3d_K(L,x)K,L) \\
    \leq \polylog(\nicefrac{R}{r}, \nicefrac{1}{t}) ~ \poly(\enc{B})~
    C^n ~ 14^n G(d_K(L,x)K,L)
    \label{eq:cvp-alg-1}
  \end{multline}
  Next since $t \leq \omega \leq t \frac{R}{r}$, we see that we will
  execute the enumeration algorithm at most $\ln_2 \frac{R}{r} + 1$
  times. Now remembering that $t = \frac{\length{z'-x}}{R}$, we see
  that all the lattice points of $L$ generated by the algorithm lie
  inside a ball of radius at most $3~\frac{R}{r}~\length{z'-x} \leq
  3~\frac{R}{r}~\sqrt{n}\length{B}$ around $x$. Hence, these lattices
  points as well as the number $t$ can be represented using at most
  $\poly(\enc{B}, \enc{x}, \ln(\frac{R}{r}))$ bits. Therefore, apart
  from in the enumeration algorithm, we only evaluate the weak norm
  oracle on inputs of size $\poly(\enc{B}, \ln(\frac{R}{r}), \enc{x},
  \ln \frac{1}{\eps})$ which is polynomial in the input. Finally, we
  filter the list $U_f$ into $S$, which requires exactly $2|U_f|$
  evaluations of the norm-oracle, where the cardinality of $U_f$ is
  bounded by (\ref{eq:cvp-alg-1}). Combining all of the above bounds,
  yields the desired result.
\end{proof}

Though the runtime of the Closest-Vectors algorithm cannot be bounded
bounded in general due to the $G(dK, L)$ term, its running time can be
controlled in interesting special cases. For example, if $K$ is
well-centered and $d_K(L,x) \leq \alpha \lambda_1(K,L)$, i.e. the
target point is relatively close to the lattice, then by
Lemma~\ref{lem:lambda1-bd} the main complexity term of the
Closest-Vectors algorithm on $K,L,x$ becomes
\begin{equation}
  G(d_K(L,x) K, L) \leq \left(4\left(1 + \frac{2 d_K(L,x)}{\lambda_1(K,L)}\right)\right)^n \leq \left(4 + 8\alpha\right)^n
\end{equation}
which is of order $2^{O(n)}$ when $\alpha = O(1)$. With this bound, we
recover (up to large $C^n$ factors) the running time of the AKS sieve
for exact CVP when the target point is close.

\subsection{Integer Programming}\label{sec:ip}

In this section, we present an algorithm for integer programming feasibility based on a general norm SVP solver.
Relying on the best known bounds for the flatness theorem (see Equation~\ref{eq:flatness}), we show that our algorithm
achieves a modest improvement in complexity of IP. For a brief history, the first fixed dimension polynomial time
algorithm for integer linear programming is due to Lenstra in ~\cite{lenstra83:_integ_progr_with_fixed_number_of_variab}
which achieved an essential complexity is $2^{O(n^3)}$. This was dramatically improved by Kannan in
\cite{kannan87:_minkow_convex_body_theor_and_integ_progr} reducing the complexity to $O(n^{2.5})^n$. The next
improvement is due K{\"o}ppe and Hildebrand \cite{arxiv/HildebrandK10} reducing the complexity to $O(n^2)^n$ while
generalizing to feasible regions defined by quasi-convex polynomials. Here we present an algorithm which runs in
$O(n^{\nicefrac{4}{3}}\log^{O(1)} n)^n$, for feasible regions equipped with a strong separation oracle (see Definition
\ref{def:strong-sep}). 

Let $f^*(n)$ denote the optimal function for the flatness theorem. Our main result here is as follows:

\begin{theorem}[Integer Programming] 
\label{thm:ip}
Let $K \subseteq R B_2^n$ be a convex body given by a strong separation oracle $SEP_K$. Let $L \subseteq \R^n$ be a
$n$-dimensional lattice given by a basis $B \in \Q^{n \times n}$. Then there exists an algorithm which either decides
that $K \cap L = \emptyset$, or returns a point $x \in K \cap L$ in expected time
\[
O(f^*(n))^n \poly(\enc{R}, \enc{a_0}, \enc{B})
\]
\end{theorem}

Unfortunately, the algorithm described above is not agnostic to the
value of $f^*$, its exact value (or any known upper bound) is needed
in the code of to guarantee the algorithm's correctness. Hence using
the best known bounds on $f^*(n)$ (see \cite{Bana99,R00}), we get an
algorithm of essential complexity $(c_1 n^{\nicefrac{4}{3}} \log^{c_2}
n)^n$ for absolute constants $c_1,c_2$.

We give an outline of the algorithm. The algorithm works as almost all previous IP algorithms do, i.e. by finding a
``thinnest'' width direction of $K$ with respect to $L$. More precisely, we adopt a recursive solution strategy, where
given $K$ and $L$ as above, we seek to find a small collection of parallel hyperplanes $H_k$, $k \in A$, such that if $K
\cap L \neq \emptyset$ then for some $k \in A$ we have that $K \cap L \cap H_k \neq \emptyset$. At this point, we simply
solve the integer program with respect to $K \cap H_k$, $L \cap H_k$ recursively for each $k \in A$, and decide that $K
\cap L$ is empty if all the subproblems return empty and return any found lattice point otherwise.  As we will explain
below, finding the above set of hyperplanes reduces to solving a shortest vector problem with respect to a general norm,
in particular the ``width'' norm of $K$, i.e. $\|x\|_{(K-K)^*} = \sup_{y \in K} \pr{y}{x} - \inf_{y \in K} \pr{y}{x}$.
In previous IP algorithms, the alluded to SVP problem is solved only approximately via a reduction to $\ell_2$ (i.e. via
an ellipsoidal approximation of the norm). The main source of improvement for our algorithm comes from the fact the we
solve the associated SVP exactly using a general norm SVP solver.

\begin{proof}[Proof of Theorem \ref{thm:ip}]
\hspace{1em} 

\paragraph{IP ALGORITHM:}

\paragraph{\emph{Basis Refinement:}}

As a first step, we will reduce to working with a lattice admitting a basis of length at most $2\sqrt{n} R$. This will
allow us to control the encoding length of the basis after each recursive invocation of the IP algorithm. To begin, we
use the MV algorithm for CVP to compute a closest vector $p \in L$ to $a_0$ in the $\ell_2$ norm. If $\|p-a_0\|_2 > R$
we declare that $K \cap L = \emptyset$ (since $K \subseteq a_0 + RB_2^n$). Otherwise, we again use the MV algorithm to
compute linearly independent lattice vectors $v_1,\dots,v_n$ achieving the successive minima of $L$, i.e. where
$\|v_i\|_2 = \lambda_i(L)$. Both invocations of the MV algorithm here take at most $2^{O(n)}\poly(\enc{B})$ time.
Letting $v_0 = 0$, compute the largest index $k$, $0 \leq k \leq n$, such that $\|v_k\| \leq 2R$. Now let $L' = L \cap
\mathrm{span}(v_0, v_1,\dots,v_k)$. 

\paragraph{Claim:} $L \cap a_0 + RB_2^n \subseteq p + L'$.
\begin{proof}
Take $y \in L \cap a_0 + RB_2^n$. Since $p \in a_0 + RB_2^n$, we have that $\|y-p\|_2 \leq 2R$. Assume that $y-p \notin
\mathrm{span}(v_0,v_1,\dots,v_k)$. Since $y-p \in L$ and $y-p$ is linearly independent from $v_0,v_1,\dots,v_k$, we get
that $\lambda_{k+1}(L) \leq 2R$. But by our choice of $k$, we know that $\lambda_{k+1}(L) > 2R$, a clear contradiction.
Therefore $y-p \in L \cap \mathrm{span}(v_0,v_1,\dots,v_k) \cap L = L'$, as needed.
\end{proof}

Since $K \subseteq a_0 + RB_2^n$, from the above claim we get that it suffices to check whether $K-p \cap L' =
\emptyset$ to solve IP feasibility with respect to $K$ and $L$. Now using standard techniques (Chris: reference needed),
we may compute a basis $B'$ for $L'$ using $v_0,v_1,\dots,v_k$ sastifying $\|B'\|_2 \leq \sqrt{k}\|v_k\|_2 \leq
2\sqrt{k}R$ in polynomial time. Let $W = \mathrm{span}(L')$ denote the linear span of $L'$, $a_0'$ denote the orthogonal
projection of $a_0-p$ onto $W$, and let $R' = \sqrt{R^2-\|a_0-a_0'\|^2}$. It is easy to check that $K-p \cap W$ is
$(a_0',R')$-circumscribed in $W$. Given that may restrict our attention to points in $K-p \cap L'$, for the rest of the
algorithm we replace $L$ by $L'$, $K$ by $K-p \cap W$ (for which a strong separation oracle is readily available via
Lemma \ref{lem:sep-slice}), and $(a_0,R)$ by $(a_0',R')$.

\paragraph{\emph{Localizing $K$:}} For the next step, we compute a strong enough ellipsoidal approximation of $K$ to begin
inferring about how $K$ interacts with $L$. To do this, we use algorithm GLS-Round (Theorem \ref{thm:gls-round}),
running against $K$ with parameter $\eps = \left(\frac{1}{4n}\right)^n\det(L)$ to deterministically compute an ellipsoid
$E+t$ such that either (1) $\vol(E) \leq \eps$ (i.e. $E$ is tiny compared to the `sparsity' of $L$), or (2) $E$
sandwiches $K$ well, i.e. $t + \frac{1}{\sqrt{n}(n+1)} E \subseteq K \subseteq t + E$. This step can be done in
$\poly(n, \log \nicefrac{R}{\eps}) = \poly(n, \enc{R}, \enc{\det(L)})$ time.

\paragraph{\emph{Branching on a ``thinnest'' width direction of $K$:}}  Here we wish to find a dual vector $y \in L^*$,
such that there exists a small number of hyperplanes of the form $H_k = \set{x: \pr{x}{y} = k}$, $k \in A \subseteq \Z$,
with the property that if $K$ contains a point of $L$ then there exists a lattice point in $H_k \cap K \cap L \neq
\emptyset$ for some $k \in A$. At this point, as explained previously, we recurse on $K \cap H_k$, $L \cap H_k$, for
all $k \in A$.  To implement such a recursive call for a specific $H_k$, $k \in A$, we compute a basis for $L \cap
H_0$ and a point $p \in H_k \cap L$, and call the IP procedure on $K-p \cap H_0$ and $L \cap H_0$. All the preprocessing
here can be done in polynomial time via standard methods, where as above we note that a strong separation oracle for
$K-p \cap H_0$ is readily computable via Lemma \ref{lem:sep-slice}.

Now to find such a $y$ and set $A$, we proceed as follows. If we are in case (1) above, we use the MV algorithm to
compute a vector $y \in \SVP(E^*,L^*)$, which can be done in $2^{O(n)} \poly(\enc{B})$ time. Noting that $(E-E)^* =
\nicefrac{1}{2}E^*$, we see that
\[
\vol((E-E)^*) = \left(\frac{1}{2}\right)^n \vol(E^*) = \left(\frac{1}{2}\right)^n \vol(B_2^n)^2 \frac{1}{\vol(E)} 
> \left(\frac{1}{2n}\right)^n \frac{1}{\vol(E)}
\]
Given that $\vol(E)^{\nicefrac{1}{n}} \leq \eps = \frac{1}{4n} \det(L)^{\nicefrac{1}{n}}$, from the above we see
that 
\[
\vol((E-E)^*)^{\nicefrac{1}{n}} > \frac{1}{2n} \frac{1}{\vol(E)^{\nicefrac{1}{n}}} \geq 2
\frac{1}{\det(L)^{\nicefrac{1}{n}}} = 2 \det(L^*)^{\nicefrac{1}{n}}
\]
Since $(E-E)^* = \nicefrac{1}{2}E^*$ is centrally symmetric, by Minkowski's first theorem (Theorem
\ref{thm:minkowski-first}) we have that $2\|y\|_{E^*} = \|y\|_{(E-E)^*} = \lambda_1((E-E)^*, L^*) < 1$. We remember that
$\|y\|_{(E-E)^*} = \sup_{x \in E} \pr{y}{x} - \inf_{x \in E} \pr{y}{x}$ is the width of $E$ with respect to $y$. Since
$y \in L^*$ we note that for any $x \in E+t \cap L$, we must have that $\pr{x}{y} \in (\pr{y}{t} + [\inf_{x \in E}
\pr{y}{x}, \sup_{x \in E} \pr{y}{x}] )\cap \Z$. Since $E$ has width $< 1$ with respect to $y$, it is easy to see that if
$E+t \cap L$ is non-empty then all the lattice points in $E+t \cap L$ must lie on the hyperplane $H = \set{x \in \R^n:
\pr{x}{y} = \round{t}}$. Since $K \subseteq E+t$, it is also clearly the case that $K \cap L \subseteq H \cap L$. To
finish with this case, we now recursively solve the Integer Program with respect $K \cap H$ and $L \cap H$, returning
empty iff $K \cap L \cap H = \emptyset$. 

If we are in case $(2)$, we know $K$ is well-sandwiched by $E$, i.e. $t + \frac{1}{\sqrt{n}(n+1)}E \subseteq K \subseteq
t + E$. To find a thin direction for $K$, we shall compute $y \in \SVP((K-K)^*,L^*,1)$. To do this, we must build a
weak distance oracle for $(K-K)^*$. Given that $K$ is well sandwiched by $E$, using the Ellipsoid Method (theorem
\ref{thm:convex-opt}), for any $y \in \Q^n$ and $\eps > 0$, we may compute $l,u \in \Q$ satisfying
\[
l - \nicefrac{\eps}{2} \leq \inf_{x \in K} \pr{y}{x} \leq l \quad u \leq \sup_{x \in K} \pr{y}{x} \leq u +
\nicefrac{\eps}{2}
\]
in polynomial time. We note now that 
\[
|\|y\|_{(K-K)^*} - (u-l)| = |\sup_{x \in K} \pr{y}{x} - \inf_{x \in K} \pr{y}{x} - (u-l)| \leq \eps
\]
as needed. Next, the SVP algorithm needs sandwiching guarantees on $(K-K)^*$. Given our guarantees on $K$, we see that
$\frac{1}{2}E^* = (E-E)^* \subseteq K-K \subseteq \frac{1}{2}(n+1)\sqrt{n}E^*$. Technically, the algorithm
Shortest-Vectors requires the sandwiching ratio with respect to euclidean balls, but this type of sandwiching is
equivalent to ellipsoidal sandwiching after linear transformation. Having constructed a weak distance oracle for
$(K-K)^*$ and computed the sandwiching guarantees, we may now call Shortest-Vectors($(K-K)^*, L^*, 1$) (Theorem
\ref{thm:svp-alg}) and retrieve $y \in L^*$ from the output. Since the sandwiching guarantees are polynomial in $n$ and
the required accuracy is $O(1)$, this call be executed in expected time $2^{O(n)} \poly(\enc{B}, \enc{E})$ time. Using
the Ellipsoid Method (theorem \ref{thm:convex-opt}) as above, we compute bounds $u, l \in \Q$ satisfying $u \leq \sup_{x
\in K} \pr{y}{x} \leq u + 1$ and $l - 1 \leq \inf_{x \in K} \pr{y}{x} \leq l$ in polynomial time. Now compute $A = [l -
1, \min \set{u + 1, l + f^*(n) + 1}] \cap \Z$. We now show that it suffices to restrict our attention to the hyperplanes
$H_k = \set{x \in \R^n: \pr{x}{y} = k}$ for $k \in A$. 

\paragraph{Claim:} If $K \cap L \neq \emptyset$, then there exists $x \in K \cap L$ such that $\pr{y}{x} \in A$.
\begin{proof}

First, if $l + f^*(n) + 1 \geq \sup_{x \in K} \pr{y}{x}$, then by our guarantees on $u$ and $l$ we have that $A
\supseteq [\inf_{x \in K} \pr{y}{x}, \sup_{x \in K} \pr{y}{x}] \cap \Z$. Since $\pr{y}{x} \in \Z$ for any $x \in L$, we
clearly have that $x \in K \cap L \Rightarrow \pr{x}{y} \in A$. Next, if $l + f^*(n) + 1 \leq \sup_{x \in K} \pr{y}{x}$,
then we have that
\[
f^*(n) \leq \sup_{x \in } \pr{y}{x} - l - 1 \leq \|y\|_{(K-K)^*} - 1 \leq \lambda_1((K-K)^*,L^*)
\]
by our assumption that $y \in \SVP((K-K)^*,L^*,1)$. Take $x_0 \in \argmin_{x \in K} \pr{y}{x}$ and examine the convex body 
\[
\tilde{K} = \left(1-\frac{f^*(n)+1}{\|y\|_{(K-K)^*}}\right)x_0 + \left(\frac{f^*(n)+1}{\|y\|_{(K-K)^*}}\right)K \text{.}
\]
Since $x_0 \in K$ and $f^*(n)+1 \leq \|y\|_{(K-K)^*}$ we get by convexity that $\tilde{K} \subseteq K$. Furthermore,
we can see that 
\[
\tilde{K} \subseteq \set{z \in \R^n: \inf_{x \in K} \pr{y}{x} \leq \pr{z}{y} \leq \inf_{x \in K} \pr{y}{x} + f^*(n) + 1}
\text{.}
\]
Therefore any $x \in \tilde{K} \cap L$ must satisfy $\pr{x}{y} \in A$. Hence if $\tilde{K} \cap L \neq \emptyset$, since
$\tilde{K} \subseteq K$ there exists $x \in K \cap L$ such that $\pr{y}{x} \in A$. We now show that $\tilde{K} \cap L
\neq \emptyset$ to complete the claim.

By homogeneity, we see that
\begin{align*}
\lambda_1((\tilde{K}-\tilde{K})^*,L^*) &= \lambda_1\left(\left(\frac{f^*(n)+1}{\|y\|_{(K-K)^*}}~(K-K)\right)^*,L^*\right) 
               = (f^*(n)+1) ~ \frac{\lambda_1((K-K)^*,L^*)}{\|y\|_{(K-K)^*}} \\
	       &\geq (f^*(n)+1) \frac{\|y\|_{(K-K)^*}-1}{\|y\|_{(K-K)^*}} \geq (f^*(n)+1) \frac{f^*(n)}{f^*(n)+1} = f^*(n)
\end{align*}
Applying the flatness theorem to $\tilde{K}$, we now get that $\mu(\tilde{K},L) \leq 1$ and hence that $\tilde{K} \cap L
\neq \emptyset$ as needed.
\end{proof}

Given the claim, we may complete the algorithm by recursively solving the integer programs with respect to $K \cap H_k$ and
$L \cap H_k$, for all $k \in A$. We return EMPTY if all calls return EMPTY, and return any found lattice point
otherwise.

\paragraph{RUNTIME:}
The correctness of the algorithm has already been discussed above, so it only remains to check that the runtime of the
algorithm is bounded by $O(f^*(n))^n \poly(\enc{a_0}, \enc{R}, \enc{B})$ on expectation (we note that the only source of
randomness in the algorithm comes from the calls to the Shortest-Vectors algorithm). The algorithm above is recursive,
where at each node of the recursion we perform the $3$ named procedures above and then break the problem into at most
$\ceil{f^*(n)} + 2$ subproblems which we solve recursively (the calls to IP on $K \cap H_k$, $L \cap H_k$ as above). Now
if we can show that the processing at each recursive node takes at most expected $2^{O(n)} \poly(\enc{a_0}, \enc{R},
\enc{B})$ time - where $a_0,R,B$ are the \emph{original} parameters provided to the top level call of the IP algorithm -
then by solving a standard recurrence relation we get that the whole running time is indeed $O(f^*(n))^n
\poly(\enc{a_0},\enc{R},\enc{B})$ on expectation as needed. 

Let us examine a specific recursion node with associated convex body $\bar{K}$, $(\bar{a}_0,\bar{R})$-circumscribed in
$\R^{\bar{n}}$, and $\bar{n}$-dimensional lattice $\bar{L}$ with basis $\bar{B}$. Now it is straightforward to see that
at this recursion node, the amount of computation is certainly bounded by $2^{O(\bar{n})}
\poly(\enc{\bar{a}_0},\enc{\bar{R}},\enc{\bar{B}})$ on expectation, since the above procedures only make calls to
subroutines with either polynomial runtimes (such as the GLS-Round algorithm, the Ellipsoid Method, and standard linear
algebraic procedures) or single exponential runtimes (such as the MV algorithm and the Shortest-Vectors algorithm). The
main issue is therefore whether the lattice basis and affine subspace passed to the next level recursion nodes have bit
size bounded by a fixed polynomial (i.e. whose degree does not depend on $n$) in the size of the original parameters.
For clarity, we only sketch the argument here. The main reason this is true is because of the Basis Refine step. Most
crucially, after the refine step, we end up with a lattice basis whose length is bounded by $2\sqrt{\bar{n}}\bar{R} \leq
2\sqrt{n}R$. Since $\bar{L}$ is a sub-lattice of our original lattice $L$, it is not hard to verify that any vector of
$\bar{L}$ (and in fact of $L$) of length less than $2 \sqrt{n} R$ has bit size bounded by $\poly(\enc{R},\enc{B})$ (for
a fixed polynomial). Hence the Basis Refine step ``smooths'' any incoming basis and subspace to ones whose bit
description is well bounded by the original parameters. Since the bit description of the lattice basis and subspace
passed to the next child node is only a fixed polynomial larger than that of the ``smoothed'' basis after our refine
step, the claim follows. The runtime is therefore bounded by $O(f^*(n))^n \poly(\enc{a_0},\enc{R},\enc{B})$ on
expectation as desired.
\end{proof}


\section{Acknowledgments}

We gratefully thank Bo'az Klartag, Gideon Schechtman, Daniele
Micciancio, Oded Regev, and Panagiotis Voulgaris for fruitful
discussions and critical ideas.  In particular, Klartag suggested to
us that the techniques of~\cite{K06} could be used for an algorithmic
construction of an M-ellipsoid, and Schechtman suggested the use of
parallelepiped tilings to construct an explicit covering.


\bibliographystyle{alphaabbrvprelim}
\bibliography{lattices,acg,cg}

\appendix

\section{M-Ellipsoid Proofs}
\label{sec:m-ellipsoid-proofs}

Here we prove correctness of the all the main M-ellipsoid algorithms
from Section~\ref{sec:m-ellipsoid-covering}.  We rely heavily on
several geometric estimates, which are listed and proved in
Section~\ref{sec:geometric-estimates} below, and on standard
algorithms from convex optimization and convex geometry, which are
described in Section~\ref{sec:standard-algorithms}.

\begin{proof}[Proof of Theorem~\ref{thm:m-ellipsoid} (Correctness of
  M-Ellipsoid)]
  Here we give more detail as to the implementation of each of the
  steps of Algorithm~\ref{alg:M-Ellipsoid}:
  \begin{itemize}
  \item Step $1$: Make a direct call to algorithm Estimate-Centroid
    (Lemma \ref{lem:estimate-centroid}) on $K$.
  \item Step $2$: If Estimate-Centroid returns an estimate $b$ of
    $b(K)$, we have the guarantee that
    \begin{equation}
      b + \frac{r}{2(n+1)\sqrt{n}}B_2^n \subseteq K \subseteq b + 2R
    \end{equation}
    Since the guarantees about $b$ in $K$ are polynomial in the input,
    we can build a weak membership oracle $O_{K-b}$ for $K-b$, where
    $K-b$ is $(0, \frac{r}{2(n+1)\sqrt{n}},2R)$-centered, in
    polynomial time from $O_K$. Now we run the algorithm of
    Theorem~\ref{thm:m-gen} on the oracle $O_{K-b}$ and retrieve the
    tentative $M$-ellipsoid $E(A)$ of $K$.
  \item Step $3$: Here we make a direct call to the algorithm
    Build-Cover on ($K$, $E(A)$) where we ask whether $N(K,E(A)) >
    (13e)^n$.
  \item Step $4$: First, we implement a weak membership oracle
    $O_{(K-K)^{*}}$ for $(K-K)^*$ from $O_K$ using the ellipsoid
    algorithm, where we can guarantee that $(K-K)^*$ is
    $(0,r\frac{1}{2R}, \frac{1}{2r})$-centered. To do this, we note
    that
    \[
    x \in (K-K)^* \Leftrightarrow \sup_{y \in K} \pr{y}{x} - \inf_{y
      \in K} \pr{y}{x} \leq 1
    \]
    Hence we can build a weak membership oracle for $(K-K)^*$ by
    approximately maximizing and minimizing with respect to $x$ over
    $K$. This can readily be done via the ellipsoid algorithm (see
    Theorem~\ref{thm:convex-opt}). The guarantees we get on $(K-K)^*$
    are seen as follows:
    \[
    a_0 + rB_2^n \subseteq K \subseteq RB_2^n \Rightarrow 2rB_2^n
    \subseteq K-K \subseteq 2RB_2^n \Rightarrow \frac{1}{2R}B_2^n
    \subseteq (K-K)^* \subseteq \frac{1}{2r}B_2^n
    \]
    Next, we note that $E(A)^* = E(A^{-1})$, and hence can be computed
    in polynomial time. Next, we call the algorithm Build-Cover on
    ($(K-K)^*$,$E(A)^*$) where we ask whether $N((K-K)^*,E(A)^*) > (25
    e \cdot 13)^n$.
  \end{itemize}

  \paragraph{Correctness:} We must show that if the algorithm
  succeeds, returning the ellipsoid $E(A)$, that $E(A)$ indeed
  satisfies
  \begin{equation}
    N(K,E) \leq \left(\sqrt{8\pi e} \cdot 13e \right)^n \quad 
    N(E,K) \leq \left(\sqrt{8\pi e} \cdot 25e \cdot 13 \cdot 289 \right)^n
  \end{equation}
  These guarantees depend only on the correctness of the algorithm
  Build-Cover. In, step $3$, if the test passes, we are guaranteed to
  get a covering $T$ of $K$ by $E$ where $|T| \leq
  \left(\sqrt{8\pi e} \cdot 13 e\right)^n$. Hence the first
  requirement is met. In step $4$, if the test passes, we are
  guaranteed that $N((K-K)^*,E^*) \leq \left(4 \sqrt{\frac{\pi e}{2}}
    \cdot 25 e \cdot 13 \right)^n$. Now by
  Theorem~\ref{thm:dual-entr}, since $E^*$ is centrally symmetric, for
  $n$ large enough, we have that
  \begin{equation}
    N(E,K) \leq 289^n N((K-K)^*,E^*) \leq \left(\sqrt{8 \pi e} \cdot 25e \cdot 13 \cdot 289 \right)^n
  \end{equation}
  as needed.

  \paragraph{Runtime:} We note that of each the steps $1-4$ already
  have a running time bounded by the desired runtime.  Hence, it
  suffices to show that the main loop is executed on expectation only
  $O(1)$ times. To do this, we first condition on the event that in
  step $1$, the returned estimate $b$ satisfies that $b-b(K) \in
  \frac{1}{n+1}E_K$. This occurs with probability at least
  $1-\frac{1}{n}$. Next, in step $2$, given that $K-b$ satisfies the
  conditions of Theorem~\ref{thm:m-gen}, i.e. that $b(K-b) = b(K)-b
  \in \frac{1}{n+1}E_K$, we may condition on the event that the
  returned ellipsoid $E(A)$ satisfies
  \begin{equation}
    N(K,E(A)) \leq (13e)^n \quad N(E(A),K) \leq (25e)^n \text{.}
  \end{equation}
  Since this event occurs with probability $1-\frac{3}{n}$, our total
  success probability is $1-\frac{4}{n}$.  Now in step $3$, given that
  $N(K,E(A)) \leq (13e)^n$, the test is guaranteed to pass. Since $E$
  is centrally symmetric, for $n$ large enough, we have that
  \begin{equation}
    N((K-K)^*,E^*) \leq (12(1+o(1)))^n N(E,K) \leq (25e \cdot 13)^n \text{.}
  \end{equation}
  Therefore, the test in step $4$ is also guaranteed to
  succeed. Finally, we see that the probability that each execution of
  the loop terminates successfully is at least $1-\frac{4}{n}$,
  therefore the expected number of runs of the loop is $O(1)$ as
  needed.
\end{proof}


\begin{proof}[Proof of Theorem~\ref{thm:m-gen} (Correctness of M-Gen)]
  The proof has two parts, first building the right oracle, then using
  it to sample and estimate the inertial ellipsoid.

  \paragraph{Building a membership oracle for the polar:}
  We first show that a polynomial time weak membership oracle for $S =
  n\left(\conv \set{K,-K}\right)^*$ can be built from $O_K$. We note
  that
  \begin{equation}
    v \in n\left(\conv \set{K,-K}\right)^* \Leftrightarrow \max ~ \set{\sup_{x \in K} \pr{v}{x}, \sup_{x \in K}
      \pr{-v}{x}} \leq n 
  \end{equation}
  Given the guarantees on $O_K$, we have that
  \begin{equation}
    \frac{n}{R}B_2^n \subseteq n\left(\conv \set{K,-K}\right)^* \subseteq \frac{n}{r}B_2^n
  \end{equation}
  Constructing a weak membership oracle for $S$ therefore requires
  only the ability to perform $2$ different approximate optimizations
  over $K$. This can achieved using the standard optimization
  techniques described in Theorem~\ref{thm:convex-opt}. Hence, a
  polynomial time weak membership oracle for $S$ can be built as
  claimed.

  \paragraph{Building the M-ellipsoid:} Le $\pi_S$ denote the uniform
  distribution on $S$. Equipped with a weak membership oracle for $S$,
  we may use the sampling algorithm of
  Theorem~\ref{thm:sampling-tech}, to sample a point $Y \in S$ with
  distribution $\sigma$ satisfying $\TVD(\sigma,\pi_S) \leq
  \frac{1}{n}$ in time $\poly(n) \polylog(\nicefrac{R}{r}, n)$.  Set
  $s = Y$, where $Y$ is the computed sample. We shall use $s$ to
  specify a reweighting of the uniform distribution on $K$.  Let
  $f_s(x) = e^{\pr{s}{x}}$ for $x \in K$ and $0$ otherwise. Using the
  algorithm described by Corollary~\ref{lem:estimate-covariance}, we
  may compute a matrix $A \in \R^{n \times n}$ satisfying
  \begin{equation}
    e^{-\nicefrac{1}{n}}E_{f_s} \subseteq E(A) \subseteq e^{\nicefrac{1}{n}} E_{f_s}
    \label{eq:cov-approx}
  \end{equation}
  with probability $1-\frac{1}{n}$ in time
  $\poly(n)\polylog(\nicefrac{R}{r})$. We return the ellipsoid
  $\sqrt{n}E(A)$ as our candidate $M$-ellipsoid for $K$.

  \paragraph{Analysis:} We now show that for $n$ large enough, the
  ellipsoid returned by this algorithm satisfies with high probability
  the covering conditions
  \begin{equation}
    N(K,\sqrt{n}E(A)) \leq (13e)^n \quad N(\sqrt{n}E(A),K)) \leq (25e)^n
  \end{equation}
  First, we condition on the event \eqref{eq:cov-approx}, i.e. that we
  get a good estimate of $E_{f_s}$. Hence at this point, our success
  probability is at least $1-\frac{1}{n}$.

  Let $\eta > 0$ be a constant to be decided later. Let $X$ be
  uniformly distributed on $S$, and let $Y$ denote the approximately
  uniform sample the above algorithm computes on $S$, remembering that
  $S = n\left(\conv\set{K,-K}\right)^*$.  Given the guarantee that
  $b(K) \in \frac{1}{n+1}E_K$, from Lemma~\ref{lem:exp-slice} setting
  $\eps = 1$, for $n$ large enough we have that
  \begin{equation}
    \E[L_{f_X}^{2n}] \leq \left((1+o(1)) ~ \sqrt{\frac{2}{\pi e}} ~ \frac{e^{\eps}}{\sqrt{\eps}} \right)^{2n}
    \leq \left((1+\eta) ~ \sqrt{\frac{2e}{\pi}} \right)^{2n}
  \end{equation}
  Using Markov's inequality, we see that
  \begin{equation}
    \Pr\left[L_{f_X} > (1+\eta)^2 \sqrt{\frac{2e}{\pi}}\right] \leq \frac{\E[L_{f_X}^{2n}]}{\left((1+\eta)^2
        \sqrt{\frac{2e}{\pi}}\right)^{2n}} \leq \frac{1}{(1+\eta)^{2n}}.
  \end{equation}

  Now since $\TVD(X,Y) \leq \frac{1}{n}$, we see that
  \begin{equation}
    \Pr\left[L_{f_Y} > (1+\eta)^2 \sqrt{\frac{2e}{\pi}}\right] \leq \frac{1}{(1+\eta)^{2n}} + \frac{1}{n} \leq \frac{2}{n}
    \label{eq:bad-slice}
  \end{equation}
  for $n$ large enough ($\eta$ will be chosen to be constant). Hence
  after additionally conditioning on the complement of event
  $\ref{eq:bad-slice}$, our success probabiblity is at least
  $1-\frac{3}{n}$. At this point, letting $s = Y$, we see that $s$
  specifies a density $f_s$ on $K$ satisfying
  \begin{equation}
    L_{f_s} \leq (1+\eta)^2 \sqrt{\frac{2e}{\pi}}.
  \end{equation}
  Furthermore since $s \in n\left(\conv \set{K,-K}\right)^*$, $b(K)
  \in \frac{1}{n+1}E_K$ and $E_K \subseteq K$, we have that
  \begin{equation}
    \frac{\sup_{x \in K} f_s(x)}{f_s(b(K))} = \sup_{x \in K} e^{\pr{s}{x-b(K)}} = \sup_{x \in K} e^{\pr{s}{x} +
      \pr{-s}{b(K)}} \leq e^{n+1}.
  \end{equation}
  Hence by Lemma~\ref{lem:iner-to-m}, letting $\sqrt{n}E(A) = T$, and
  $\delta = e^{\nicefrac{1}{n}}$, we get that
  \begin{equation}
    N(K, \sqrt{n}E(A)) \leq (12\delta)^n ~ \frac{4}{3} ~ \frac{\sup_{x \in K} f_s(x)}{f_s(b(K))} \leq
    12^n e ~ \frac{4}{3} ~ e^{n+1} \leq (12e(1+\eta))^n
  \end{equation}
  and
  \begin{align}
    \begin{split}
      N(\sqrt{n}E(A), K) &\leq (12\delta^2)^n ~ \vol(\sqrt{n}B_2^n) ~ \frac{4}{3} ~ L_{f_s}^n \\
      &\leq 12^n e^2 ~ (\sqrt{2\pi e}(1+o(1)))^n ~ \frac{4}{3} ~
      \left((1+\eta)^3\sqrt{2}\right)^n \leq (24e(1+\eta)^3)^n
    \end{split}
  \end{align}
  for $n$ large enough. Choosing $\eta > 0$ such that $(1+\eta)^3 =
  \nicefrac{25}{24}$ yields the result.
\end{proof}


\begin{proof}[Proof of Theorem~\ref{thm:build-cover} (Correctness of Build-Cover)]
  The goal here is to either compute a covering of $K$ by $E$, or
  conclude that $N(K,E)$ is large. To make this task easier, we will
  replace $E$ by a parallelepiped $P$ inscribed in $E$, and use a
  tiling procedure (since $P$ can be used to tile space) to cover
  $K$. We will show any cover produced in this way is not much larger
  than $N(K,E)$, and hence will help provide a lower bound on
  $N(K,E)$. Furthermore since $P \subseteq E$, any cover of $K$ by $P$
  immediately translates into a cover of $K$ by $E$.

  \paragraph{Building $P$:}
  To compute $P$ we will need to perform some standard matrix
  algebra. First we compute the Cholesky Factorization of $A$, i.e. we
  compute $V \in \R^{n \times n}$ such that $A = V^t V$. Next we
  compute $B = V^{-1}$, the inverse of $V$, and label the columns of
  $B$ as $B = (b_1,\dots,b_n)$. Both of the computations here can be
  done in time $\poly(\enc{A})$ via standard methods. Now we note that
  \begin{equation}
    \pr{b_i}{b_j}_A = b_i^tAb_j = (B^tAB)_{ij} = (V^{-t}V^tVV^{-1})_{ij} = (\Id_n)_{ij}.
  \end{equation}
  Hence the vectors $(b_1,\dots,b_n)$ form an orthonormal basis of
  with respect to the dot product $\pr{\cdot}{\cdot}_A$.  Therefore
  the ellipsoid $E(A)$ may be expressed as
  \begin{equation}
    E(A) = \set{x \in \R^n: x^tAx \leq 1} = \set{x \in \R^n: \sum_{i=1}^n \pr{b_i}{x}_A^2 \leq 1}.
  \end{equation}
  Now define $P$ as
  \begin{equation}
    P = \left\{x \in \R^n: |\pr{b_i}{x}_A| \leq \frac{1}{\sqrt{n}}\right\} 
    = \left\{\sum_{i=1}^n a_ib_i: |a_i| \leq \frac{1}{\sqrt{n}}~,~ 1 \leq i \leq m \right\}
  \end{equation}
  where the second equality follows since the $b_i$s are orthonormal
  under $\pr{\cdot}{\cdot}_A$. Now for $x \in \R^n$, we see that
  \begin{equation}
    \max_{1 \leq i \leq n} |\pr{b_i}{x}_A| 
    \leq \left(\sum_{i=1}^n \pr{b_i}{x}_A^2\right)^{1/2} \leq \sqrt{n} \max_{i \leq i
      \leq n} |\pr{b_i}{x}_A|  \quad \Rightarrow \quad P \subseteq E(A) \subseteq \sqrt{n}P.
  \end{equation}
  Now a standard computation yields that
  \begin{align}
    \vol(P) = \left(\frac{2}{\sqrt{n}}\right)^n
    \det(A)^{-1/2} \quad \vol(E(A)) =
    \left(\frac{\sqrt{2\pi e}(1+o(1))}{\sqrt{n}}\right)^n
    \det(A)^{-1/2}
    \label{eq:vol-est}
  \end{align}
  where we remember here that $\det(B) = \det(V^{-1}) = \det(V)^{-1} =
  \det(A)^{-1/2}$.  Therefore we have that $\vol(E(A))
  \leq \left(\sqrt{\frac{\pi e}{2}}(1+o(1))\right)^n \vol(P)$.

\paragraph{Tiling $K$ with $P$:}
Define the lattice
\begin{equation}
  L = \left\{ \sum_{i=1}^n \frac{2}{\sqrt{n}} z_i b_i: z_i \in \Z, ~ 1 \leq i \leq m\right\} \text{,}
\end{equation}
so $L$ is the lattice spanned by the vectors
$\frac{2}{\sqrt{n}}(b_1,\dots,b_n)$. From here it is straightforward
to verify that $P$ tiles space with respect to $L$, so $L + P = \R^n$
and for $x,y \in L$, $x \neq y$, $x + \mathrm{int}(P) \cap y +
\mathrm{int}(P) = \emptyset$, i.e. the interiors are disjoint. In
fact, one can see that $P$ is simply a shift of the fundamental
parallelepiped of $L$ with respect to the basis
$\frac{2}{\sqrt{n}}(b_1,\dots,b_n)$.

We now wish to tile $K$ with copies of $P$. To do this we examine the
set $H = \set{x \in L: x + P \cap K \neq \emptyset}$. Since $P + L =
\R^n$, it is easy to see that
\begin{equation}
  K \subseteq H + P. 
\end{equation}
Hence, we shall want to decide for $x \in L$, whether $x + P \cap K
\neq \emptyset$. Since we only have a weak membership oracle for $K$,
we will only be able to decide whether $x + P$ approximately
intersects $K$. To formalize this, we build an weak intersection
oracle $\mathrm{INT}$ which queried on $x \in \R^n$, $\eps > 0$
satisfies
\begin{equation}
  \mathrm{INT}(x,\eps) = \begin{cases} 0:& \quad x + P \cap K = \emptyset \\ 1:& \quad x + (1+\eps)P \cap K \neq
    \emptyset \end{cases}.
\end{equation}
Using this oracle we will be able to overestimate $T$, and compute a
set $S \subseteq L$ such that
\begin{equation}
  H \subseteq S \subseteq \set{x: x + (1+\eps)P \cap K \neq \emptyset}
\end{equation}
which will suffice for our purposes. Now to build $\mathrm{INT}$, we
first remark that for $x \in \R^n$, $t \geq 0$
\begin{equation}
  x + tP \cap K \neq \emptyset \Leftrightarrow \inf_{y \in K} \|y-x\|_P \leq t \Leftrightarrow \inf_{y \in K} \sqrt{n}
  \max_{1 \leq i \leq n} |\pr{b_i}{y-x}_A| \leq t.
\end{equation}
Hence deciding the minimum scaling $t$ of $P$ for which $x + tP \cap K
\neq \emptyset$ is equivalent to solving a simple convex program. The
above convex program is exactly in the form described in
Theorem~\ref{thm:convex-opt}, hence for $\eps > 0$, and $x \in \Q^n$,
we may compute a number $\omega \geq 0$ such that
\begin{equation}
  |\omega - \inf_{y \in K} \|y-x\|_K| \leq \eps
  \label{eq:approx-int}
\end{equation}
in time
$\poly(n,\enc{x},\enc{A})\polylog(\nicefrac{R}{r},\nicefrac{1}{\eps})$. We
now build $\mathrm{INT}$. On query $x \in \Q^n$, $\eps > 0$, we do the
following:
\begin{enumerate}
\item Compute $\omega \geq 0$ satisfying $|\omega - \inf_{y \in K}
  \|y-x\|_K| \leq \frac{\eps}{2}$.
\item If $\omega \leq 1+\frac{\eps}{2}$ return $1$, otherwise return
  $0$.
\end{enumerate}
From $(\ref{eq:approx-int})$ the above procedure clearly runs in
polytime. To prove correctness, we must show that
$\mathrm{INT}(x,\eps) = 1$ if $x + P \cap K \neq \emptyset$ and
$\mathrm{INT}(x,\eps) = 0$ if $x + (1+\eps)P \cap K = \emptyset$. If
$x + P \cap K \neq \emptyset$, we note that $\inf_{y \in K} \|y-x\|_K
\leq 1$, hence by the guarantee on $\omega$ we have that
\begin{equation}
  \omega \leq \inf_{y \in K} \|y-x\|_K + \frac{\eps}{2} \leq 1 + \frac{\eps}{2}, 
\end{equation}
and so we correctly classify $x$. If $x + (1+\eps)P \cap K =
\emptyset$, then $\inf_{y \in K} \|y-x\|_K > 1+\eps$ and so
\begin{equation}
  \omega \geq \inf_{y \in K} \|y-x\|_K - \frac{\eps}{2} > 1 + \frac{\eps}{2}
\end{equation}
as needed.

We now compute a tiling of $K$. The idea here is simple. We define a
graph $G$ on the lattice $L$, where for $x,y \in L$, $x \sim y$ iff
$x-y \in \frac{2}{\sqrt{n}}\set{\pm b_1,\dots, \pm b_n}$. We identify
each lattice point $x \in L$ with the tile $x + P$. Starting from the
tile centered at $0$, we begin a breadth first search on $G$ of the
tiles intersecting $K$. In this way, we will compute the connected
component containing $0$ in $G$ of tiles intersecting $K$.  Lastly, if
the number of intersecting $K$ tiles exceeds $\left(4\sqrt{\frac{\pi
      e}{2}}H\right)^n$, we abort and return that $N(K,E) \geq
H^n$. The algorithm is given in Algorithm~\ref{alg:tiling}.

\begin{algorithm}
  \caption{Computing a tiling.}
  \label{alg:tiling}
  \begin{algorithmic}[1]
    \STATE $M \leftarrow \set{0}, N \leftarrow \set{0}, T \leftarrow
    \emptyset$.  \WHILE{$N \neq \emptyset$} \STATE choose $x \in N$
    \STATE $N \leftarrow N \setminus \set{c}$ \IF{$\mathrm{INT}(x,
      \nicefrac{1}{n}) = 1$} \STATE $T \leftarrow T \cup \set{x}$
    \IF{$|T| > \left(4\sqrt{\frac{\pi e}{2}}H\right)^n$} \RETURN FAIL
    \ENDIF
    \FORALL{$\delta \in \frac{2}{\sqrt{n}}\set{\pm b_1,\dots, \pm
        b_n}$} \IF{$x + \delta \notin M$} \STATE $N \leftarrow N \cup
    \set{x}$, $M \leftarrow M \cup \set{x}$
    \ENDIF
    \ENDFOR
    \ENDIF
    \ENDWHILE
    \RETURN $T$
  \end{algorithmic}
\end{algorithm}

\paragraph{Correctness:} To argue correctness of the above algorithm,
we must guarantee that the algorithm either computes a valid covering
of $K$ or that it proves that $N(K,E) > H^n$. For $\eps \geq 0$, let
\begin{equation}
  H_{\eps} = \set{x \in L: x + (1+\eps)P \cap K \neq \emptyset} \quad \text{ and } \quad 
  H'_{\eps} = \set{x \in L: \mathrm{INT}(x,\eps) = 1}
\end{equation}
From the description above, we see that the algorithm performs a
breath first search on $G$ starting from $0$ of the tiles in
$H'_{\nicefrac{1}{n}}$. From the properties of the weak intersection
oracle $\mathrm{INT}$, we know that $H_0 \subseteq
H'_{\nicefrac{1}{n}} \subseteq H_{\nicefrac{1}{n}}$.

The goal of the algorithm is to discover a super-set of $H_0$. Since
$H_0 \subseteq H'_{\nicefrac{1}{n}}$, the algorithm will correctly add
elements of $H_0$ to the cover $T$ if it finds them. Since we perform
a breadth first search from $0$, to guarantee we find all of $H_0$ we
need only ensure that $H_0$ forms a connected subgraph of $G$. As
noted before, the set of tiles indexed by $H_0$ are just lattice
shifts of the fundamental parallelepiped of $L$ with respect to the
basis $\frac{2}{\sqrt{n}}(b_1,\dots,b_n)$. In this setting, the
connectivity of $H_0$ with respect the edges defined by the basis
(i.e. the set of tiles touching any convex set), is a classical
fact. Therefore, the algorithm will indeed discover all of $H_0$,
provided that the partial cover $T$ remains no larger than
$\left(4\sqrt{\frac{\pi e}{2}}H\right)^n$.

Now we must justify that if the algorithm aborts, i.e. if $|T| >
\left(4\sqrt{\frac{\pi e}{2}}H\right)^n$, that indeed $N(K,E) >
H^n$. Now at every timestep we have that $T \subseteq
H'_{\nicefrac{1}{n}} \subseteq H_{\nicefrac{1}{n}}$.  Therefore, to
show correctness, it suffices to show that $|H_{\nicefrac{1}{n}}| \leq
\left(4\sqrt{\frac{\pi e}{2}}\right)^n N(K,E)$. Now for $x \in
H_{\nicefrac{1}{n}}$, we have that
\begin{equation}
  x + (1+\nicefrac{1}{n}) P \cap K \neq \emptyset \Rightarrow x \in K + (1+\nicefrac{1}{n})P 
  \Rightarrow x + P \in K + (2+\nicefrac{1}{n})P
\end{equation}
Furthermore, since for $x,y \in H_{\nicefrac{1}{n}}$, $x \neq y$, $x +
\mathrm{int}(P) \cap y + \mathrm{int}(P) = \emptyset$, we have that
\begin{equation}
  \vol(K + (2+\nicefrac{1}{n})P) \geq \vol(\cup_{x \in H_{\nicefrac{1}{n}}} x + P) = |H_{\nicefrac{1}{n}}|\vol(P)
\end{equation}
Using that $P \subseteq E$, and $\vol(E) \leq \left(\sqrt{\frac{\pi
      e}{2}}(1+o(1))\right)^n \vol(P)$ we get
\begin{align}
  \begin{split}
    |H_{\nicefrac{1}{n}}| &\leq \frac{\vol(K +
      (2+\nicefrac{1}{n})P)}{\vol(P)}
    \leq \left(\sqrt{\frac{\pi e}{2}}(1+o(1))\right)^n \frac{\vol(K + (2+\nicefrac{1}{n})E)}{\vol(E)} \\
    &\leq \left(\sqrt{\frac{\pi
          e}{2}}(1+o(1))(3+\nicefrac{1}{n})\right)^n N(K,E) \leq
    \left(4\sqrt{\frac{\pi e}{2}}\right)^n N(K,E)
  \end{split}
\end{align}
for $n$ large enough. Hence the algorithm correctly decides whether
$N(K,E) > H^n$.

\paragraph{Runtime:} The running time of the algorithm is proportional
to the number of tiles visited and the number of edges crossed during
the search phase. Since all the tiles visited in the algorithm are
adjacent to the tiles in the set $T$, and the number of edges is $2n$,
the total number of tiles visited is at most $2n |T| \leq 2n
\left(4\sqrt{\frac{\pi e}{2}}H\right)^n$. Furthermore, the edges
traversed correspond to all the outgoing edges from $T$, and hence is
bounded by the same number. Now at every visited tile, we make a call
to $\mathrm{INT}(x,\nicefrac{1}{n})$ for some $x \in L$, which takes
$\poly(n,\enc{A}) \polylog(\nicefrac{R}{r})$ time.  Hence the total
running time is
\begin{equation}
  \poly(n,\enc{A})\polylog(\nicefrac{R}{r})\left(4\sqrt{\frac{\pi e}{2}}H\right)^n
\end{equation}
as needed.
\end{proof}

\subsection{Geometric Estimates}
\label{sec:geometric-estimates}

Here we list and prove  the necessary geometric inequalities that
we used in the proofs above.  We begin with a slight extension of
Theorem~\ref{thm:symmetrize}.

\begin{theorem} 
\label{thm:symmetrize-ext}
Let $K$ be a convex body such that $b(K) \in t E_K$, for some $t \in [0,1)$. Then
\begin{equation}
\vol(K \cap -K) \geq \left(\frac{1-t}{2}\right)^n  \vol(K)
\end{equation}
\end{theorem}

\begin{proof}
From Theorem~\ref{thm:symmetrize} we have that
\begin{equation}
   \frac{1}{2^n} \vol(K) \leq \vol(K-b(K) \cap -K+b(K)) = \vol(K \cap -K + 2b(K))
\end{equation}
Next, we note that for $x \in \R^n$
\begin{equation}
K \cap -K + 2x \neq \emptyset \Leftrightarrow 2x \in K+K \Leftrightarrow x \in K
\label{eq:symm-support}
\end{equation}

Since $b(K) \in tE_K$ and $b(K) + E_K \subseteq K$, we see that
$(1-t)E_K \subseteq K$. Hence we can write 
\begin{equation}
0 = t(-2nb(K)) + (1-t) 2b(K) \text{,}
\end{equation}
where $-nb(K) \in -(1-t)E_K = (1-t)E_K \subseteq K$. Now we see that
\begin{equation}
t\left(K \cap (-K + -2nb(K))\right) + 
(1-t)\left(K \cap (-K + 2b(K))\right) \subseteq K \cap -K
\end{equation}
where both sets on the left hand side are non-empty by \eqref{eq:symm-support}. Therefore by the Brunn-Minkowski
inequality, we have that 
\begin{align}
\begin{split}
\vol(K \cap -K)^{\frac{1}{n}} &
	\geq t\vol\left(K \cap (-K + -n2b(K))\right)^{\frac{1}{n}} + 
             (1-t)\vol\left(K \cap (-K + 2b(K))\right)^{\frac{1}{n}} \\
	&\geq (1-t)\vol\left(K \cap (-K + 2b(K))\right)^{\frac{1}{n}} 
         \geq \frac{1-t}{2} \vol(K)^{\frac{1}{n}}
\end{split}
\end{align}
Therefore we get that
\[
\vol(K \cap -K) \geq \left(\frac{1-t}{2}\right)^n \vol(K)
\]
as needed.
\end{proof}

The next lemma is a slight specialization of~\cite[Theorem
5]{MP00}. We require this inequality for the M-ellipsoid certification
procedure.
\begin{theorem}[Duality of Entropy]
\label{thm:dual-entr}
Let $K,T \subseteq \R^n$ be convex bodies where $T$ is centrally symmetric. Then
\begin{equation}
  N(T,K) \leq \left((1+o(1))288\right)^n \cdot N((K-K)^*,T^*)
\end{equation}
and
\begin{equation}
	N((K-K)^*,T^*) \leq \left(12(1+o(1))\right)^n \cdot N(T,K).
\end{equation}
\end{theorem}

\begin{proof}
Since the above quantities are invariant under shifts of $K$, we may shift $K$ so that $b(K)=0$. Applying Theorem
\ref{thm:symmetrize}, we see that that $\vol(K-K) \leq 4^n \vol(K) \leq 8^n \vol(K \cap -K)$, where we note that since
$0 \in K$ we have that $K \cap -K \subseteq K \subseteq K-K$.  Next applying the covering estimates from Lemma~\ref{lem:cov-est}, we get that 
\[
N(K-K,K) \leq N(K-K,K \cap -K) \leq 3^n \frac{\vol(K-K)}{\vol(K \cap -K)} \leq 24^n.
\]
From here, we see that
\begin{equation}
N(T,K) \leq N(T,K-K)N(K-K,K) \leq 24^n N(T,K-K).
\end{equation}
Next since both $T$ and $K-K$ are centrally symmetric, we apply Lemma~\ref{lem:cov-est} to get that
\[
N(T,(K-K)) \leq 3^n \frac{\vol(T)}{\vol((K-K) \cap T)}.
\]
Now we note that $((K-K) \cap T)^* = \conv \set{(K-K)^*, T^*}$. Hence applying the Blashke-Santal{\'o} inequality to
$\vol(T)$ and the Bourgain-Milman inequality to $\vol((K-K) \cap T)$ we get that
\[
3^n \frac{\vol(T)}{\vol((K-K) \cap T)} \leq (6(1+o(1)))^n \frac{\vol(\conv \set{(K-K)^*, T^*})}{\vol(T^*)}
\]
Since $0$ is both in $(K-K)^*$ and $T^*$, we see that $\conv \set{(K-K)^*, T^*)} \subseteq (K-K)^* + T^*$ and hence 
\[
(6(1+o(1)))^n \frac{\vol(\conv \set{(K-K)^*, T^*})}{\vol(T^*)} \leq (6(1+o(1)))^n \frac{\vol((K-K)^* + T^*)}{\vol(T^*)}.
\]
Lastly, applying Lemma \ref{lem:cov-est} to the last estimate, we get that
\[
(6(1+o(1)))^n \frac{\vol((K-K)^* + T^*)}{\vol(T^*)} \leq (12(1+o(1)))^n  N((K-K)^*, T^*).
\]
Combining the above estimates yields the first desired inequality.

Now switching the roles $(K-K)$ and $T$ with $(K-K)^*$ and $T^*$, we have that
\[
N((K-K)^*,T^*) \leq (12(1+o(1))^n N(T,K-K) \leq (12(1+o(1))^n N(T,K),
\]
yielding the second inequality.
\end{proof}

We now make precise the relationship between the isotropic constant of the exponential reweightings defined by Klartag
\cite{K06} and the M-ellipsoid.

\begin{lemma}
  \label{lem:iner-to-m}
  Let $K \subseteq \R^n$ be a convex body. Take $s \in \R^n$ and let
  $f_s(x) = e^{\pr{s}{x}}$ for $x \in K$ and $0$ otherwise. Let $T
  \subseteq \R^n$ be a convex body such that for some $\delta \geq 1$
  we have that
  \begin{equation}
    \frac{\sqrt{n}}{\delta} E_{f_s} \subseteq T \subseteq \delta \sqrt{n} E_{f_s}
  \end{equation}
  where $E_{f_s}$ is the inertial ellipsoid of $f_s$. Then we have
  that
  \begin{equation}\label{m-and-iso}
    N(K, T) \leq (12\delta)^n ~ \frac{4}{3} ~ \frac{\sup_{x \in K} f_s(x)}{f_s(b(K))} 
    \quad \text{ and } \quad 
    N(T, K) \leq (12\delta^2)^n ~ \vol(\sqrt{n}B_2^n) ~ \frac{4}{3} ~ L_{f_{s}}^n
  \end{equation}
  where $b(K)$ is the centroid of $K$, and $L_{f_s}$ is the
  isotropic constant of $f_s$.
\end{lemma}
\begin{proof}
  Since the above estimates are all invariant under shifts of $K$, we
  may assume that $b(f_s) = 0$ (centroid of $f_s$).  We note that
  $b(f_s) \in K$ always and hence $0 \in K$. Let $X$ be distributed as
  $\pi_{f_s}$, where $\pi_{f_s}$ is the probability measure induced by
  $f_s$. So we have that $\E[X] = b(f_s) = 0$ and $\E[XX^t] =
  \cov(f_s)$.

  \paragraph{} Remember that $E_{f_s} = \set{x: x^t\cov(f_s)^{-1}x
    \leq 1}$, therefore $\|x\|_{E_{f_s}} = \sqrt{x^t \cov(f_s)^{-1}
    x}$. Now note that
  \begin{align}
    \E[\|X\|_{E_{f_s}}^2] &= \E[X^t\cov(f_s)^{-1}X] =
    \E[\mathrm{trace}[\cov(f_s)^{-1}XX^t] ]
    = \mathrm{trace}[\cov(f_s)^{-1}\E[XX^t]] \\
    &= \mathrm{trace}[\cov(f_s)^{-1} \cov(f_s)] =
    \mathrm{trace}[\Id_n] = n.
  \end{align}
  Now by Markov's inequality, we have that
  \begin{equation}
    \pi_{f_s}(2\sqrt{n}E_{f_s}) = 1-\Pr[\|X\|_{E_{f_s}} > 2\sqrt{n}]
    \geq 1-\frac{\E[\|X\|_{E_{f_s}}^2]}{4n} = 1 - \frac{n}{4n} = \frac{3}{4}.
    \label{eq:ell-tail}
  \end{equation}
  By Jensen's inequality, we see that
  \begin{equation}
    \int_K f_s(x) dx = \int_K e^{\pr{s}{x}}dx = \vol(K) \int_K e^{\pr{s}{x}} \frac{dx}{\vol(K)} \geq \vol(K)
    e^{\pr{s}{b(K)}} = \vol(K) f_s(b(K)),
    \label{eq:rint-lb}
  \end{equation}
  where $b(K)$ is the centroid of $K$.

  Using \eqref{eq:rint-lb} and \eqref{eq:ell-tail} we see that
  \begin{equation}
    \vol(2\sqrt{n}E_{f_s} \cap K) \geq \frac{\int_{2\sqrt{n}E_{f_s}}
      f_s(x)dx}{\sup_{x \in K} f_s(x)} \geq \frac{3}{4} ~ \frac{\int_K f_s(x) dx}{\sup_{x \in K} f_s(x)} \geq \frac{3}{4}
    ~ \frac{f_s(b(K))}{\sup_{x \in K} f(x)} \vol(K) .
    \label{eq:ell-int-lb}
  \end{equation}

  Using that $\frac{\sqrt{n}}{\delta}E_{f_s} \subseteq T$, $0 \in K$,
  $\delta \geq 1$, and by \eqref{eq:ell-int-lb} we get that
  \begin{align}
    \begin{split}
      \vol(T \cap K) &\geq \vol\left(\frac{\sqrt{n}}{\delta} E_{f_s}
        \cap K\right) = \left(\frac{1}{\delta}\right)^n
      \vol(\sqrt{n}E_{f_s} \cap \delta K)
      \geq \left(\frac{1}{\delta}\right)^n \vol\left(\sqrt{n}E_{f_s} \cap \frac{1}{2}K\right)  \\
      &= \left(\frac{1}{2\delta}\right)^n \vol(2\sqrt{n}E_{f_s} \cap
      K) \geq \left(\frac{1}{2\delta}\right)^n ~ \frac{3}{4} ~
      \frac{f_s(b(K))}{\sup_{x \in K} f(x)} ~ \vol(K).
    \end{split}
    \label{eq:itm-1}
  \end{align}

  Using the definition of $L_{f_s}$, \eqref{eq:ell-tail},
  $\sqrt{n}E_{f_s} \subseteq \delta T$ and that $0 \in K$, we get that
  \begin{align}
    \begin{split}
      \det(\cov(f_s))^{\frac{1}{2}} &= L_K^n ~ \frac{\int_K f_s(x)
        dx}{\sup_{x \in K} f_s(x)} \leq L_K^n ~ \frac{4}{3} ~
      \frac{\int_{2\sqrt{n}E_{f_s}} f_s(x) dx}{\sup_{x \in K} f_s(x)}
      \leq L_K^n ~ \frac{4}{3} ~ \vol(2\sqrt{n}E_{f_s} \cap K) \\
      &\leq L_K^n ~ \frac{4}{3} ~ \vol(2 \delta T \cap K) \leq
      (2\delta L_K)^n ~ \frac{4}{3} ~ \vol(T \cap K).
    \end{split}
    \label{eq:itm-2a}
  \end{align}

  Using that $T \subseteq \delta \sqrt{n}E_{f_s}$ and the ellipsoid
  volume formula \eqref{eq:ell-vol-form}, we have that
  \begin{equation}
    \vol(T) \leq \vol(\delta \sqrt{n}E_{f_s}) = \delta^n \vol(\sqrt{n}B_2^n) \det(\cov(f_s))^{\nicefrac{1}{2}}.
    \label{eq:itm-2b}
  \end{equation}
  Combining equations \eqref{eq:itm-2a},\eqref{eq:itm-2b} we get
  that
  \begin{equation}
    \vol(T) \leq (2\delta^2 L_K)^n ~ \vol(\sqrt{n}B_2^n) ~ \frac{4}{3} ~ \vol(T \cap K).
    \label{eq:itm-2}
  \end{equation}

  Now applying Lemma \ref{lem:cov-est} to the inequalities \eqref{eq:itm-1},\eqref{eq:itm-2} the theorem
  follows.
\end{proof}

From Lemma \ref{lem:iner-to-m}, we see that if the slicing conjecture is true, then for any convex body, its inertial
ellipsoid appropriately scaled is an $M$-ellipsoid. To bypass this, Klartag shows that for any convex body $K$, there
exists a ``mild'' exponential reweighting $f_s$ of the uniform density on $K$ with bounded isotropic constant. As one
can see from Lemma \ref{lem:iner-to-m}, the severity of the reweighting controls $N(K,\sqrt{n}E_{f_s})$ whereas the
isotropic constant of $f_s$ controls $N(\sqrt{n}E_{f_s},K)$. 

The main tool to establish the existence of ``good'' exponential
reweightings for $K$ is the following lemma, which one can extract
from the proof of Theorem \ref{thm:eps-slicing} in \cite{K06}. We will
use it here for $\eps = 1$, in which case the expectation below is of
order $2^{O(n)}$. The argument is essentially identical to that
of~\cite{K06}; we include it for completeness.

\begin{lemma}[\cite{K06}]
  \label{lem:exp-slice}
  Let $K \subseteq \R^n$ be a convex body such that $b(K) \in
  \frac{1}{n+1}E_K$. For $s \in \R^n$, let $f_s:K \rightarrow \R^+$
  denote the function $f_s(x) = e^{\pr{s}{x}}$, $x \in K$. Let $X$ be
  distributed as $\eps n \left(\conv \set{K, -K}\right)^*)$ for some
  real $\epsilon > 0$. Then we have
  \begin{equation}
    \E[L_{f_X}^{2n}] \leq \left((1+o(1)) ~ \sqrt{\frac{2}{\pi e}} ~
      \frac{e^{\eps}}{\sqrt{\eps}} \right)^{2n}
  \end{equation}
\end{lemma}

\begin{proof}
  For $s \in \R^n$ define $f_s:K \rightarrow \R_+$ by $f_s(x) =
  e^{\pr{s}{x}}$ for $x \in K$. In Lemma 3.2 of~\cite{K06} is it shown
  that
  \begin{equation}
    \int_{\R^n} \det(\cov(f_s)) ds = \vol(K)
    \label{lem:transport}
  \end{equation}

  By Theorem \ref{thm:sandwich}, we have that $E_K + b(K) \subseteq
  K$. Since $b(K) \in \frac{1}{n+1}E_K$ by assumption, we see that
  $\frac{n}{n+1}E_K \subseteq E_K + b(K) \subseteq K$. Hence $0 \in
  K$. From~\cite{RS58}, we know that for any convex body $K$ such that
  $0 \in K$, we have that $\vol(\conv \set{K,-K}) \leq 2^n \vol(K)$.

  Let $L = \conv \set{K,-K}$. Note that
  \begin{equation}
    L^* = \left(\conv \set{K,-K}\right)^* = \set{y: |\pr{x}{y}| \leq 1, ~ \forall x \in K}
  \end{equation}
  Since $L$ is centrally symmetric by the Bourgain-Milman inequality
  (Theorem~\ref{thm:vol-prod}), we have that
  \begin{equation}
    \vol(L^*)\vol(L) \geq \left((1+o(1))\frac{\pi e}{n}\right)^n
  \end{equation}
  Hence we get that
  \begin{equation}
    \vol(L^*) \geq \left(\frac{(1+o(1))\pi e}{n\vol(L)^{\frac{1}{n}}}\right)^n
    \geq \left(\frac{(1+o(1))\pi e}{2n\vol(K)^{\frac{1}{n}}}\right)^n
    \label{eq:slice-1}
  \end{equation}
  Take $s \in \eps n L^*$. We examine the properties of $f_s:K
  \rightarrow \R_+$. Since $s \in \eps n L^*$, we see that
  \begin{equation}
    \sup_{x \in K} f_s(x) = e^{\sup_{x \in K} \pr{s}{x}} \leq e^{\eps n}
  \end{equation}
  Since $b(K) \subseteq \frac{1}{n+1}E_K \subseteq \frac{1}{n}K$ and
  $s \in \eps n \left(\conv \set{K,-K}\right)^*$, we see that
  $|\pr{s}{b(K)}| \leq \eps$. Now by Jensen's inequality, we have that
  \begin{align}
    \begin{split}
      \int_K e^{\pr{s}{x}} dx &= \vol(K) \left(\int_K e^{\pr{s}{x}}
        \frac{dx}{\vol(K)}\right)
      \geq \vol(K) e^{\int_K \pr{s}{x} \frac{dx}{\vol(K)}} \\
      &= \vol(K) e^{\pr{s}{b(K)}} \geq \vol(K) e^{-\eps}
    \end{split}
  \end{align}
  Now we see that
  \begin{equation}
    L_{f_s}^{2n} = \left( \sup_{x \in K} \frac{f_s(x)}{\int_K f_s(x) dx} \right)^2 \det(\cov(f_s))
    \leq \left( \frac{e^{\eps n}}{\vol(K) e^{-\eps}} \right)^2 \det(\cov(f_s))
    = \frac{e^{2(n+1) \eps}}{\vol(K)^2} \det(\cov(f_s))
    \label{eq:slice-2}
  \end{equation}
  Applying inequality (\ref{eq:slice-2}), Lemma 3.2 of~\cite{K06}, and
  equation (\ref{eq:slice-1}), we get that
  \begin{align}
    \begin{split}
      \frac{1}{\vol(\eps n L^*)} \int_{\eps n L^*} L_{f_s}^{2n} ds
      &\leq \frac{e^{2(n+1) \eps}}{\vol(\eps n L^*)\vol(K)^2} \int_{\eps n L^*} {\vol(K)^2}\det(\cov(f_s)) ds \\
      &\leq \frac{e^{2(n+1) \eps}}{\vol(\eps n L^*)\vol(K)^2} \vol(K)
      \leq \left(\frac{(1+o(1))e^{2\eps}}{\eps n \vol(L^*)^{\frac{1}{n}} \vol(K)^{\frac{1}{n}}}\right)^n \\
      &\leq \left(\frac{(1+o(1))2 e^{2\eps} }{\pi e \eps}\right)^n
      = \left((1+o(1)) ~ \sqrt{\frac{2}{\pi e}} ~ \frac{e^{\eps}}{\sqrt{\eps}} \right)^{2n}\\
    \end{split}
  \end{align}
  The above quantity is exactly $\E[L_{f_{X}}]$ since $X$ is uniform over $\eps n L^*$. The statement thus follows.
\end{proof}


\section{Additional Background}
\label{sec:additional-background}

For two probability distributions $\sigma_1,\sigma_2$ over a domain
$\mathcal{X}$, their \emph{total
  variation} (or \emph{ statistical}) \emph{distance} is
\begin{equation}
  \TVD(\sigma_1,\sigma_2) = \sup_{A \subseteq \mathcal{X}} |\sigma_1(A)-\sigma_2(A)|.
\end{equation}

\subsection{Logconcave functions}

We will need to work with the generalization of convex bodies to
logconcave functions.  A function $f:\R^n \rightarrow \R_+$ is
logconcave if for all $x,y \in \R^n$, and $0 \leq \alpha \leq 1$, we
have that
\begin{equation}
  f(\alpha x + (1-\alpha)y) \geq f(x)^\alpha f(y)^{1-\alpha}
\end{equation}
The canonical examples of logconcave functions are the indicator
functions of convex bodies as well as the Gaussian distributions. We
will now generalize the concepts defined before for convex bodies to
logconcave functions.

For a logconcave function $f$ on $\R^n$ such that $0 < \int_{\R^n}
f(x)\, dx < \infty$, we define the associated probability measure
(distribution) $\pi_f$, where for measurable $A \subseteq \R^n$, we
have
\begin{equation}
  \pi_f(A) = \frac{\int_A f(x)\, dx}{\int_{\R^n} f(x)\, dx}.
\end{equation}

We define the {\em centroid} (or barycenter) and {\em covariance}
matrix of $f$ as
\begin{align*}
  b(f) &= \frac{\int_{\R^n} x f(x)dx}{\int_{\R^n} f(x)\, dx} &
  \cov(f)_{ij} &= \frac{\int_{\R^n} (x_i-b(f)_i)(x_j-b(f)_j)
    f(x)\, dx}{\int_{\R^n} f(x)\, dx} ~~ 1 \leq i,j \leq n
\end{align*}

The matrix $\cov(f)$ is positive semi-definite and symmetric.  We say
that $f$ is isotropic, or in isotropic position, if $b(f) = 0$ and
$\cov(f)$ is the identity matrix.  Define the {\em inertial ellipsoid}
of $f$ as
\begin{equation*}
  E_f = E(\cov(f)^{-1}) = \set{x: x^t \cov(f)^{-1} x \leq 1}
  \label{def:iner-ell}
\end{equation*}
The {\em isotropic constant} of $f$ is defined as
\begin{equation*}
  L_f = \left(\sup_{x \in \R^n} \frac{f(x)}{\int_{\R^n} f(x)dx}
  \right)^{\nicefrac{1}{n}} \cdot \det(\cov(f))^{\nicefrac{1}{2n}}.
  \label{def:is-const}
\end{equation*}
A natural extension of the slicing conjecture
(Conjecture~\ref{conj:slicing}) is that $L_f$ is bounded by a
universal constant.  This generalized slicing conjecture was shown by
Ball~\cite{Ball88} to be equivalent to the slicing conjecture for
convex bodies, up to a constant factor in the precise bound.


For a convex body $K$, let $\pi_K$ denote the uniform measure
(distribution) over $K$. Let $f_K$ denote the associated density, i.e.,
\[
f_K(x) = \frac{1}{\vol(K)} I[x \in K],
\]
We note that the definitions coincide exactly if we replace $K$ by $f_K$,
i.e., $\cov(K) = \cov(f_K)$, $b(K) = b(f_K)$, $L_K = L_{f_K}$, etc. We
extend all the notions defined above for log-concave functions to
convex bodies in the same way, e.g. we let $E_K = E_{f_K}$.
We say that $K$ is in isotropic position
if $b(K) = 0$ and $\cov(K)$ is the identity (a
different normalization is sometimes used in asymptotic convex geometry,
namely, $b(K)=0$, $\vol(K)=1$, and $\cov(K)$ is constant diagonal).

\subsection{Computational model}
For a rational matrix $A$, we define $\enc{A}$ as the length of the binary
encoding of $A$. The lattice algorithms presented will have complexity
depending on the dimension $n$ of the lattice and the bit length of the
description of the input basis.

Since we work with general (semi-)norms, we shall need an appropriate way to represent them. We now define the three
different types of oracles that we will need. For convenience, our semi-norms will always be indexed by a convex body
$K$. With some slight modifications, we will adopt the terminology from \cite{GLS}.

Let $K \subseteq \R^n$ be a convex body. For $\eps \geq 0$, we define
\begin{equation}
  K^{\eps} = K + \eps B_2^n \quad \text{ and } \quad K^{-\eps} = \set{x \in K: x + \eps B_2^n \subseteq K}
\end{equation}

We say that $K$ is \emph{$(a_0,R)$-circumscribed} if $K \subseteq a_0 + RB_2^n$ for some $a_0 \in \Q^n$ and $R \in \Q$.
We say that $K$ is \emph{$(a_0,r,R)$-centered} if $a_0 + rB_2^n \subseteq K \subseteq a_0 + RB_2^n$ for $a_0 \in \Q^n$, $r,R
\in \Q$.  We will always assume that the above parameters are given explicitly as part of the input to our problems, and
hence our algorithms will be allowed to depend polynomially in $\enc{a_0},\enc{r},\enc{R}$.

\begin{definition}
\label{def:weak-memb}
A \emph{weak membership oracle} $O_K$ for $K$ is function which takes as input a point $x \in \Q^n$ and real $\eps > 0$,
and returns
\begin{equation}
  O_K(x,\eps) = \begin{cases} 1 &: x \in K^{\eps} \\ 0 &: x \notin K^{-\eps} \end{cases}
\end{equation}
where any answer is acceptable if $x \in K^{\eps} \setminus K^{-\eps}$.
\end{definition}

\begin{definition}
\label{def:strong-sep}
A \emph{strong separation oracle} $\mathrm{SEP}_K$ for $K$ on input $y
\in \Q^n$ either returns YES if $y \in K$, or some $c \in \Q^n$ such
that $\pr{c}{x} < \pr{c}{y}$, $\forall x \in K$.
\end{definition}

When working with the above oracle, we assume that there is a polynomial $\phi$, such that on input $y$ as above, the
output of $\mathrm{SEP}_K$ has size bounded by $\phi(\enc{y})$. The runtimes of algorithms using $\mathrm{SEP}_K$ will
therefore depend on $\phi$.

Let $K$ be a convex body containing the origin.

\begin{definition}
\label{def:weak-dist}
A \emph{weak distance oracle} $D_K$ for $K$ is a function that takes as
input a point $x \in \Q^n$ and $\eps > 0$, and returns a rational number satisfying
\begin{equation}
  \absfit{D_K(x,\eps)-\length{x}_K} \leq \eps.
\end{equation}
\end{definition}

As above, we assume the existence of a polynomial $\phi$, such that the size of the output of $D_K$ on $(x,\eps)$ is
bounded by $\phi(\enc{x},\enc{\eps})$. For a $(0,r,R)$-centered body $K$, $\forall x \in \R^n$, we crucially have that
\[
\frac{1}{R} \length{x} \leq \length{x}_K \leq \frac{1}{r} \length{x} \text{.}
\]


\subsection{Standard Algorithms}
\label{sec:standard-algorithms}

Here we list some of the algorithmic tools we will require.

The following theorem is essentially the classical equivalence between weak
membership and weak optimization~\cite{YN76,GLS}.

\begin{theorem}[Convex Optimization via Ellipsoid Method]
\label{thm:convex-opt}
Let $K \subseteq \R^n$ an $(a_0,r,R)$-centered convex body given by a weak membership oracle $O_K$. Let $A \in
\Q^{m \times n}$, $c \in \Q^m$. Define $f:\R^n \rightarrow \R$ as
\begin{equation}
  f(x) = \max_{1 \leq i \leq m} \pr{A_i}{x}+c_i
\end{equation}
where $A_i$ is the $i^{th}$ row of $A$. Then for $\eps > 0$, a number $\omega \in \Q$ satisfying
\begin{equation}
  |\omega - \inf_{x \in K} f(x)| \leq \eps
\end{equation}
can be computed using $O_K$ in time
\begin{equation}
  \poly(n,\enc{A},\enc{a_0},\enc{c}) \polylog(\nicefrac{R}{r},\nicefrac{1}{\eps})
\end{equation}
\end{theorem}

We will also need an algorithm from \cite{GLS}, which allows one to deterministically compute an ellipsoid with
relatively good ``sandwiching'' guarantees for a convex body $K$. We present a small modification of the result in GLS:

\begin{theorem}[Algorithm GLS-Round]
  \label{thm:gls-round}
  Let $K \subseteq R^n$ be an $(a_0,R)$-circumscribed convex body
  given by a strong-separation oracle $\mathrm{SEP}_K$. Then for any
  $\eps > 0$, in $\poly(\log \nicefrac{R}{\eps}, \enc{a_0} n)$ time
  one can compute $A \succ 0$, $A \in \Q^{n \times n}$ and $t \in
  \R^n$, such that the ellipsoid $E = E(A)$ satisfies $K \subseteq E +
  t$, and one of the following: (a)~$\vol(E) \leq \eps$, or
  (b)~$\frac{1}{(n+1)n^{\nicefrac{1}{2}}} E + t \subseteq K$.
\end{theorem}

The next theorem comes from the literature on random walks on convex bodies~\cite{LV3,LV06,LV2}.

\begin{theorem}[Algorithm Logconcave-Sampler,~\cite{LV06}]
  \label{thm:sampling-tech}
Let $K \subseteq \R^n$ be a $(a_0,r,R)$-centered convex body given by a weak membership oracle $O_K$. Let $f:K
\rightarrow \R_+$ be a polynomial time computable log-concave function satisfying
\begin{equation}
\sup_{x \in K} f(x) \leq \beta^{n} f(0)
\end{equation}
for some $\beta > 1$. Let $\eps,\tau > 0$. Then the following can be computed:
\begin{enumerate}
\item A random point $X \in K$ with distribution $\sigma$ satisfying $\TVD(\sigma,\pi_{f_s}) \leq \tau$ in time
\begin{equation}
  \poly(n,\enc{a_0}) \polylog(n, \nicefrac{R}{r}, \beta, \nicefrac{1}{\tau})
\end{equation}
\item A point $b \in K$ and a matrix $A \in \Q^{n \times n}$ such that $\forall ~ x \in \R^n$
\begin{equation}
|\pr{x}{b-b(f_s)}| \leq \eps~x^t \cov(f_s)x  \quad \text{ and } \quad |x^t(A-\cov(f_s))x| \leq \eps~x^t\cov(f_s)x \text{,}
\end{equation}
with probability $1-\delta$ in time
\begin{equation}
  \poly(n, \enc{a_0}, \nicefrac{1}{\eps}) \polylog(n, \nicefrac{R}{r}, \beta, \nicefrac{1}{\delta}) \text{.}
\end{equation}
\end{enumerate}
\end{theorem}

The following simple lemma allows us to construct a strong separation oracle for any hyperplane section of a convex body
already equipped with a strong separation oracle.

\begin{lemma} Let $K \subseteq \R^n$ be a convex body presented by a strong separation oracle $SEP_K$. Let $H = \set{x
\in \R^n: Ax = b}$ denote an affine subspace, where $A \in \Q^{m \times n}$, $b \in \Q^m$. Then one
can construct a separation oracle for $K \cap H$, such that on input $y \in H$, the oracle executes in time
$\poly(\enc{y}, \enc{A}, \enc{b})$ using a single call to $SEP_K$.
\label{lem:sep-slice}
\end{lemma}

\begin{proof}
  We wish to construct a strong separation oracle for $K ~\cap~ H$, where $H = \set{x \in \R^n: Ax = b}$ is an affine
  subspace, given a strong separation oracle for $K$. To do this given $y \in H$, we do the following.  First, we call
  $SEP_K$ on $y$. If $SEP_K$ returns that $y \in K$, we return YES. If $SEP_K$ returns a separator $c \in \R^n$ such that
  $\sup_{x \in K} \pr{c}{x} < \pr{c}{y}$, we compute $\bar{c}$ the orthogonal projection of $c$ onto $W = \set{x \in R^n:
    Ax = 0}$ (the lineality space of $H$). If $\bar{c} = 0$, we note that $\pr{\bar{c}}{\cdot}$ is constant over $H$.
  Therefore if $K \cap H \neq \emptyset$, there exists $x \in K \cap H \subseteq K$ such that $\pr{c}{x} = \pr{c}{y}$, a
  contradiction.  Hence if $\bar{c} = 0$, we return that $K \cap H$ is EMPTY. Otherwise, we simply return $\bar{c}$. Since
  the derived oracle simply calls $SEP_K$ once and projects any found separator onto the lineality space of $H$, the
  runtime is clearly $\poly(\enc{A},\enc{b},\enc{y})$ as needed.
\end{proof}

We now derive some straightforward applications of the above
fundamental tools.

\begin{corollary}[Algorithm Estimate-Covariance]
\label{lem:estimate-covariance}
Let $K \subseteq \R^n$ be an $(a_0,r,R)$-centered convex body given by a weak membership oracle $O_K$.
Let $f:K \rightarrow \R_+$ be a polynomial time computable log-concave function satisfying
\begin{equation}
\sup_{x \in K} f(x) \leq e^{2n} f(0).
\end{equation}
Then an ellipsoid $E(A)$, $A \in \Q^{n \times n}$, can be computed satisfying
\begin{equation}
    e^{-\nicefrac{1}{n}} E_{f_s} \subseteq E(A) \subseteq e^{\nicefrac{1}{n}} E_{f_s}
\end{equation}
with probability $1-\delta$ in time $\poly(n,\enc{a_0}, \log (\tfrac{R}{r}), \log(\tfrac{1}{\delta}))$.
\end{corollary}

\begin{proof}
  Using Theorem~\ref{thm:sampling-tech}, we can compute a matrix $B \subseteq \Q^{n \times n}$ satisfying
  \begin{equation}
    |x^t(B-\cov(f_s))x| \leq \frac{1}{n}~x^t\cov(f_s)x \quad \forall ~ x \in \R^n \text{,}
    \label{eq:ec-1}
  \end{equation}
  with probabiliy $1-\delta$ in time $\poly(n) \polylog(n, \nicefrac{R}{r}, \nicefrac{1}{\delta})$. We now condition on
  the event (\ref{eq:ec-1}). Remembering that $x^tBx = \|x\|_B^2$ and
  $x^t\cov(f_s)x =
  \|x\|_{\cov(f_s)}^2$, we may rewrite (\ref{eq:ec-1}) as
  \begin{equation}
    \sqrt{\frac{n-1}{n}}\|x\|_{\cov(f_s)} \leq \|x\|_B \leq \sqrt{\frac{n+1}{n}}\|x\|_{\cov(f_s)}
  \end{equation}
  From the above, we see that the ellipsoid $E(\cov(f_s)) = \set{x: \|x\|_{\cov(f_s)} \leq 1}$ and $E(B) = \set{x:
    \|x\|_B \leq 1}$ satisfy
  \begin{equation}
    \sqrt{\frac{n}{n+1}}E(\cov(f_s)) \subseteq E(B) \subseteq \sqrt{\frac{n}{n-1}}E(\cov(f_s))
    \label{eq:ell-con1}
  \end{equation}
  Remembering that the polar ellipsoids satisfy 
  \begin{equation}
    E(B)^* = E(B^{-1}) \quad \text{ and } \quad E(\cov(f_s))^{-1} = E(\cov(f_s)^{-1}) = E_{f_s} \text{.}
  \end{equation}
  where the last equality follows by the definition of $E_{f_s}$. Taking the polars of the above ellipsoids, the
  containment relationships in (\ref{eq:ell-con1}) flip, and we get
  \begin{equation}
    \sqrt{\frac{n-1}{n}}E_{f_s} \subseteq E(B^{-1}) \subseteq \sqrt{\frac{n+1}{n}}E_{f_s}
    \label{eq:ell-con2}
  \end{equation}
  Now using the inequalities $1-\frac{1}{n} \geq e^{-\nicefrac{2}{n}}$ for $n \geq 3$ and $1+\frac{1}{n} \leq
  e^{\nicefrac{2}{n}}$, we see that (\ref{eq:ell-con2}) implies
  \begin{equation}
    e^{-\nicefrac{1}{n}} E_{f_s} \subseteq E(B^{-1}) \subseteq e^{\nicefrac{1}{n}} E_{f_s}
  \end{equation}
  as needed. Letting $A = B^{-1}$, the ellipsoid $E(A)$ satisfies the desired requirements.
\end{proof}

\begin{corollary}[Algorithm Estimate-Centroid]
  \label{lem:estimate-centroid}
  There is a probabilistic algorithm Estimate-Centroid that, given a
  $(0,r,R)$-centered convex body $K$ presented by a weak membership
  oracle $O_K$ and some $\delta > 0$, in time $\poly(n) \polylog(n,
  \nicefrac{R}{r}, \nicefrac{1}{\delta})$ either outputs FAIL (with
  probability at most $\delta$) or some $b \in K$ such that:
  \begin{equation}
    b + \frac{r}{2(n+1)\sqrt{n}}B_2^n \subseteq K \subseteq b + 2RB_2^n
  \end{equation} 
  and with probability at least $1-\delta$,
  \begin{equation}
    \displaystyle b-b(K) \in \frac{1}{n+1}E_{K}.
  \end{equation}
\end{corollary}

\begin{proof}
  Using Theorem~\ref{thm:sampling-tech}, we compute a center $b \in K$ satisfying
  \begin{equation}
    |\pr{x}{b-b(K)}| \leq \frac{1}{(n+1)^2} ~x^t \cov(K)x \quad \forall ~ x \in \R^n \text{,} 
    \label{eq:ec-2}
  \end{equation}
  with probability $1-\delta$ in time $\poly(n) \polylog(n, \nicefrac{R}{r}, \nicefrac{1}{\delta})$.

  First, check whether
  \begin{equation}
    O_K\left(b \pm \frac{3r}{4(n+1)} e_i, \frac{r}{4(n+1)\sqrt{n}}\right) = 1 \quad \text{ for } 1 \leq i \leq n
    \label{eq:well-guar-1}
  \end{equation}
  If any of the above tests fail, abort and return FAIL.  

  Let $\delta = \frac{r}{n+1}$. If these tests pass, by the properties of $O_K$ we know that
  \begin{align}
    \begin{split}
      b + \frac{3\delta}{4} \conv \set{\pm e_1, \dots, \pm e_n} \subseteq K^{\frac{\delta}{4\sqrt{n}}} 
      &\Rightarrow b + \frac{3\delta}{4\sqrt{n}} B_2^n \subseteq K^{\frac{\delta}{4\sqrt{n}}} 
      \Rightarrow b + \frac{\delta}{2\sqrt{n}} B_2^n \subseteq K
    \end{split}
  \end{align}
  Since $b \in K \subseteq RB_2^n$, we clearly also have that $K \subseteq b + 2RB_2^n$. Hence conditioned up outputting $b$,
  we have that
  \begin{equation}
    b + \frac{r}{2(n+1)\sqrt{n}}B_2^n \subseteq K \subseteq b + 2RB_2^n
  \end{equation} 
  as needed.

  We now show that if the event (\ref{eq:ec-2}) holds, then the above test will pass and condition $(b)$ will also be
  satisfied. Since this event holds with probability $1-\delta$, this will suffice to prove the statement.

  For the center $b$, we note that for all $x \in (n+1)E(\cov(f_s))$, by equation
  (\ref{eq:ec-2}) we have that
  \begin{equation}
    |\pr{b-b(K)}{x}| \leq \frac{1}{(n+1)^2} x^t \cov(K) x \leq \frac{1}{(n+1)^2} (n+1)^2 = 1
  \end{equation}
  Therefore, we have that $b-b(K) \in \left((n+1)E(\cov(K))\right)^* = \frac{1}{n+1} E_K$ as needed.

  We now show that the tests must all pass. From Theorem~\ref{thm:sandwich}, we know that
  \begin{equation}
    b(K) + \sqrt{\frac{n+2}{n}}E_K \subseteq K \subseteq b(K) + \sqrt{n(n+2)}E_K
  \end{equation}
  By the guarantee on $O_K$, we know that $rB_2^n \subseteq b(K) + \sqrt{n(n+2)}E_K$. But we have that
  \begin{align}
    \begin{split}
      rB_2^n - b(K) \subseteq \sqrt{n(n+2)}E_K 
      &\Rightarrow rB_2^n + b(K) \subseteq \sqrt{n(n+2)}E_K  \\
      &\Rightarrow \nicefrac{1}{2}(rB_2^n - b(K)) + \nicefrac{1}{2}(rB_2^n + b(K)) \subseteq \sqrt{n(n+2)}E_K \\
      &\Rightarrow rB_2^n \subseteq \sqrt{n(n+2)}E_K
    \end{split}
  \end{align}
  since both $E_K$ and $B_2^n$ are symmetric. From the inequality $n+1 \geq \sqrt{n(n+2)}$, we have that
  \begin{equation}
    \frac{r}{n+1}B_2^n \subseteq \frac{\sqrt{n(n+2)}}{n+1}E_K \subseteq E_K
    \label{eq:guar-iner}
  \end{equation}
  Since $b-b(K) \in \frac{1}{n+1}E_K$ by assumption, and $\sqrt{\frac{n+2}{n}}E_K + b(K) \subseteq K$, we get that
  \begin{equation}
    b \in b(K) + \frac{1}{n+1}E_K \Rightarrow b + E_K \subseteq b(K) + \frac{n+2}{n+1}E_K \Rightarrow
    b + E_K \subseteq b(K) + \sqrt{\frac{n+2}{n}}E_K \subseteq K
  \end{equation}
  Therefore by~\ref{eq:guar-iner}) we have that $b+ \frac{r}{n+1}B_2^n \subseteq K$. Letting $\delta = \frac{r}{n+1}$,
  from the previous sentence we see that 
  \begin{equation}
    b \pm \frac{3}{4} \delta e_i \in K^{-\nicefrac{\delta}{4}} \subseteq K^{-\frac{\delta}{4\sqrt{n}}}
  \end{equation}
  Therefore by the properties of $O_K$, the tests in~\ref{eq:well-guar-1} must all pass. The claim thus holds.
\end{proof}

\subsection{Geometric Inequalities}

Perhaps the most fundamental inequality in the geometry of numbers is Minkowski's first theorem,
which is stated as follows:

\begin{theorem}\label{thm:minkowski-first} Let $L \subseteq \R^n$ be an $n$ dimensional lattice
and let $K \subseteq \R^n$ denote a centrally symmetric convex body. Then
\[
\lambda_1(K,L) \leq 2\left(\frac{\det(L)}{\vol(K)}\right)^{\nicefrac{1}{n}}
\]
\end{theorem}

The following gives bounds on how well the inertial ellipsoid approximates a convex body. The estimates below are from
~\cite{KLS95}:

\begin{theorem}\label{thm:sandwich}
For a convex body $K \subseteq \R^n$, the inertial ellipsoid $E_K$ satisfies
\begin{equation}
  \sqrt{\frac{n+2}{n}} \cdot E_K \subseteq K-b(K) \subseteq \sqrt{n(n+2)}
  \cdot E_K
\end{equation}
where equality holds for any simplex.
\end{theorem}

The above containment relationship was shown in~\cite{MP89} for
centrally symmetric bodies (with better bounds), and by~\cite{Son90}
for general bodies with suboptimal constants.

The next theorem gives estimates on the volume product, a fundamental
quantity in Asymptotic Convex Geometry. The upper bound for centrally
symmetric bodies follows from the work of Blashke~\cite{Bla18}, and
for general bodies by Santal{\'o}~\cite{San49}. The lower bound was
first established by Bourgain and Milman~\cite{BM87}, and was recently
refined by Kuperberg~\cite{Kup08}, as well as by Nazarov~\cite{Naz09},
where Kuperberg achieves the best constants. Finding the exact
minimizer of the volume product is a major open problem in Asymptotic
Convex Geometry.

\begin{theorem}
  \label{thm:vol-prod}
  Let $K$ be a convex body in $\R^n$. Then we have
  \begin{equation}
    \vol(B_2^n)^2 \geq \inf_{x \in K} \vol(K-x)\vol((K-x)^*) \geq \left(\frac{\pi e(1+o(1))}{2n}\right)^n \text{.}
  \end{equation}
  If $K$ is centrally symmetric, then
  \begin{equation}
    \vol(B_2^n)^2 \geq \vol(K)\vol(K^*) \geq \left(\frac{\pi e(1+o(1))}{n}\right)^n \text{.}
  \end{equation}
  In both cases, the upper bounds are equalities if and only if $K$ is
  an ellipsoid.
\end{theorem}

We remark that the upper and lower bounds match within a $4^n$ factor
($2^n$ for symmetric bodies) since $\vol(B_2^n)^2 = \left(\frac{2\pi
    e(1+o(1))}{n}\right)^n$. Using the M-ellipsoid, one can directly
derive weak bounds (i.e., with sub-optimal constants) on the volume
product. Furthermore, as we shall see in Section
\ref{sec:m-ellipsoid-proofs} , the techniques developed by
Klartag~\cite{K06} can be used to derive the existence of the
M-ellipsoid as an essential consequence of the volume product bounds.

The next theorem gives useful volume estimates for some basic
operations on a convex body. The first estimate is due to Rogers and
Shepard~\cite{RS57}, and the second is due Milman and
Pajor~\cite{MP00}:

\begin{theorem}
\label{thm:symmetrize}
Let $K \subseteq \R^n$ be a convex body. Then
\[
\vol(K-K) \leq \binom{2n}{n}\vol(K) \leq 4^n \vol(K).
\]
If $b(K) = 0$, i.e., the centroid of $K$ is at the origin, then
\[
\vol(K) \leq 2^n \vol(K \cap -K).
\]
\end{theorem}

Lastly, we relate some well-known covering estimates. Here $N(K,T) = \min \set{|\Lambda|: \Lambda \subseteq \R^n, K
\subseteq \Lambda + T}$, where $K,T$ are convex bodies in $\R^n$.
\begin{lemma}
  \label{lem:cov-est}
  Let $K,T \subseteq \R^n$ be convex bodies. Then
  \begin{equation}
    N(K,T) \leq 6^n \inf_{c \in \R^n} \frac{\vol(K)}{\vol(K \cap (T+c))}
    \quad \text{ and } \quad \frac{\vol(K+T)}{\vol(T)} \leq 2^n N(K,T).
  \end{equation}
  If $T$ is centrally symmetric, then
  \begin{equation}
   N(K,T) \leq \frac{\vol(K + T/2)}{\vol(T/2)}.
  \end{equation}
  If both $K$ and $T$ are centrally symmetric, then
  \begin{equation}
  N(K,T) \leq 3^n \frac{\vol(K)}{\vol(K \cap T)}.
  \end{equation}
\end{lemma}

\begin{proof}
  Let us first examine the case where $T$ is centrally symmetric, where we wish to show that
  \begin{equation}
    \label{eq:ce-st-pck}
    N(K,T) \leq \frac{\vol(K+T/2)}{\vol(T/2)}
  \end{equation}
  Let $\Lambda \subseteq K$ be a maximal subset of $K$ such that for
  $x_1,x_2 \in \Lambda$, $x_1 \neq x_2$, $x_1 + T/2 \cap x_2 + T/2 =
  \emptyset$.

  \paragraph{Claim 1:} $\displaystyle K \subseteq \cup_{x \in \Lambda}~ x + T~$.

  Take $y \in K$. By maximality of $\Lambda$, there exists $x \in \Lambda$ such that
  \[
  y + T/2 ~ \cap ~ x + T/2 \neq \emptyset \quad \Rightarrow \quad y
  \in x + T/2 - T/2 \quad \Rightarrow \quad y \in x + T
  \]
  where the last equality follows since $T$ is centrally
  symmetric. The claim thus follows.

  \paragraph{Claim 2:} $\displaystyle |\Lambda| \leq \frac{\vol(K+T/2)}{\vol(T/2)}$.

  For $x \in \Lambda$, note that since $x \in K$, we have
  that $x + T/2 \subseteq K + T/2$. Therefore $\Lambda + T/2 \subseteq K$.  Since
  the sets $x + T/2$, $x \in \Lambda$, are disjoint, we have that
  \begin{align}
    \vol(K+T/2) \geq \vol(\Lambda + T/2) = |\Lambda| \vol(T/2)
  \end{align}
  as needed.

  Now let us assume that $K$ is also symmetric. Since both $K$ and $T$ are symmetric, we have
  that $K \cap T$ is also symmetric. Therefore by the estimate in \eqref{eq:ce-st-pck} we get that
  \begin{equation}
    \label{eq:ce-2}
    N(K,T) \leq N(K,T \cap K) \leq \frac{\vol(K + \nicefrac{1}{2}(T \cap K))}{\vol(\nicefrac{1}{2}(T \cap K))}
    \leq \frac{\vol(\nicefrac{3}{2} K)}{\vol(\nicefrac{1}{2}(T \cap K))} = 3^n \frac{\vol(K)}{\vol(T \cap K)}
  \end{equation}
  as needed.

  Now we examine the case where neither $K$ nor $T$ is necessarily symmetric. Since the covering estimate is shift
  invariant, we may assume that $K$ and $T$ have been shifted such that $\vol(K \cap T)$ is maximized, and that the
  centroid of $K \cap T$ is at $0$. Let $S = (K \cap T) \cap -(K \cap T)$. By Theorem \ref{thm:symmetrize} we have that
  $\vol(S) \geq 2^{-n} \vol(K \cap T)$. Note that $S$ is a centrally symmetric convex body. Hence by identical reasoning as
  in \eqref{eq:ce-2} we get that
  \[
  N(K,T) \leq 3^n \frac{\vol(K)}{\vol(S)} \leq 6^n \frac{\vol(K)}{\vol(K \cap T)}
  \]
  as needed.

  Lastly, pick any $\Lambda \subseteq \R^n$ such that $K \subseteq \Lambda + T$ and $|\Lambda| = N(K,T)$. Now we see that
  \[
  \vol(K+T) \leq \vol((\Lambda + T) + T) = \vol(\Lambda + 2T) \leq |\Lambda| \vol(2T) = 2^n \vol(T) N(K,T)
  \]
  as needed.
\end{proof}


\end{document}